\documentclass[11pt]{article}

\usepackage[margin=1in]{geometry}
\usepackage{amsmath,amssymb,amsthm,mathtools}
\usepackage{bm}
\usepackage{graphicx}
\usepackage{tikz}
\usetikzlibrary{arrows.meta,positioning}
\usepackage{microtype}
\usepackage{enumitem}
\usepackage[hidelinks]{hyperref}
\hypersetup{
  pdftitle={Zero-SNR Analyticity of the Scalar MMSE Is Equivalent to Gaussianity},
  pdfauthor={Yixing Zhang},
  pdfsubject={Minimum mean-square error, Gaussian channels, and zero-SNR analyticity},
  pdfkeywords={MMSE, Gaussian channel, zero-SNR, analyticity, Borel transform, backward heat flow, Hermite polynomials}
}

\allowdisplaybreaks
\setlist[itemize]{leftmargin=1.6em,itemsep=0.25em,topsep=0.3em}
\newtheorem{theorem}{Theorem}[section]
\newtheorem{proposition}[theorem]{Proposition}
\newtheorem{lemma}[theorem]{Lemma}
\newtheorem{corollary}[theorem]{Corollary}
\theoremstyle{definition}

\theoremstyle{remark}
\newtheorem{remark}[theorem]{Remark}
\newtheorem{example}[theorem]{Example}

\newcommand{\E}{\mathbb{E}}
\newcommand{\Pp}{\mathbb{P}}
\newcommand{\R}{\mathbb{R}}
\newcommand{\C}{\mathbb{C}}

\newcommand{\mmse}{\operatorname{mmse}}
\newcommand{\Var}{\operatorname{Var}}
\newcommand{\He}{\operatorname{He}}
\newcommand{\dd}{\,\mathrm{d}}
\newcommand{\cB}{\mathcal{B}}
\newcommand{\cC}{\mathcal{C}}

\title{Zero-SNR Analyticity of the Scalar MMSE\\Is Equivalent to Gaussianity}
\author{Yixing Zhang\\[-0.15em]
\small \href{mailto:yixing.zhang.duke@gmail.com}{\texttt{yixing.zhang.duke@gmail.com}}}
\date{September 13, 2026}

\begin{document}
\maketitle

\begin{abstract}
Let $Y_s=\sqrt{s}X+Z$, where $Z$ is standard Gaussian and independent of
the real random variable $X$.  We prove that, under the square-exponential
moment condition $\E e^{\beta X^2}<\infty$ for some $\beta>0$, the scalar
minimum mean-square error $\mmse_X(s)$ is analytic at zero signal-to-noise
ratio if and only if $X$ is Gaussian, with constant random variables included
as degenerate Gaussians.

The proof converts estimation in the Gaussian channel into a backward heat
flow acting on the moment-generating function $M(z)=\E e^{zX}$.  Under the
stated tail condition, every non-Gaussian input forces $M$ to have a nonzero
complex zero.  We show that each zero cluster produces a finite singularity
in its localized Borel transform at the action $\xi=z_0^2/2$.  After removing
the action scale, the Borel coefficients have a nonzero $n^{-1/2}$ prefactor
for a simple zero.  A zero of multiplicity $m\geq2$ splits according to
the roots of a Hermite polynomial and instead contributes a prefactor
$n^{-m/2}e^{r_m\sqrt{2n}}$.  A
finite-disc localization and relative-cycle continuation argument then show
that at least one such singularity survives in the full Borel transform.
Thus, for every non-Gaussian input in the stated class, the formal zero-SNR
expansion is Gevrey--1 but divergent.  Rational-MMSE rigidity and the
analogous analyticity criterion for mutual information follow as
corollaries.
\end{abstract}

\tableofcontents

\section{Introduction}

Consider the scalar Gaussian channel
\begin{equation}
    Y_s=\sqrt{s}\,X+Z,
    \qquad Z\sim N(0,1),\quad Z\perp X,
    \label{eq:channel}
\end{equation}
and the minimum mean-square error
\begin{equation}
    \mmse_X(s)
    =\E\!\left[(X-\E[X\mid Y_s])^2\right],
    \qquad s\geq 0.
    \label{eq:mmse-def}
\end{equation}
The MMSE is a central object in nonlinear estimation and information theory.
In particular, the I--MMSE identity relates it to the derivative of mutual
information with respect to signal-to-noise ratio
\cite{GuoShamaiVerdu2005}.  Vector gradients, input-estimate
representations, mismatched analogues, and pointwise versions of this
identity were subsequently developed in
\cite{PalomarVerdu2006,PalomarVerdu2007,Verdu2010,VenkatWeissman2012}.
The monotonicity, smoothness, crossing, and higher-derivative properties of
the MMSE have consequently been studied in considerable detail
\cite{GuoWuShamaiVerdu2011,WuVerdu2012,BustinEtAl2013,Ledoux2016,Nguyen2024}.

For positive SNR, Gaussian smoothing makes $s\mapsto\mmse_X(s)$ real
analytic under mild assumptions.  The endpoint $s=0$ is different.  If all
moments of $X$ exist, the MMSE is infinitely right differentiable at zero
\cite{GuoWuShamaiVerdu2011}, but smoothness alone does not imply convergence
of its Taylor series.  A symmetric binary input is already known to be
nonanalytic at zero \cite{GuoWuShamaiVerdu2011};
Section~\ref{sec:simple} strengthens this observation by identifying the
divergent coefficient scale and its controlling complex action.  This raises
a rigidity question: which input laws have a genuinely convergent zero-SNR
expansion?

Low-SNR and wideband expansions have a long history in information theory;
see, among others,
\cite{Verdu2002,PrelovVerdu2004,PayaroPalomar2009,AlghamdiCalmon2021}.
Beyond fixed-order expansions, Ledoux's heat-flow identities
\cite{Ledoux2016} and the combinatorial analysis of Mansanarez, Poly, and
Swan \cite{MansanarezPolySwan2024} study the algebraic structure of
higher-order derivatives.  The latter work also gives sufficient conditions
for the MMSE conjecture, which asks whether the full MMSE curve determines
the input law up to translation and reflection.  Our question is different:
we ask whether analyticity at the zero-SNR endpoint alone forces Gaussianity,
and study the large-order growth of the complete expansion rather than
reconstructing a law from equality of two MMSE curves.  Our
posterior-mean representation is also related to the classical
Robbins--Miyasawa--Tweedie identities and their modern higher-order
extensions
\cite{Robbins1956,Miyasawa1961,Brown1971,Efron2011,ManorMichaeli2024}.

Our main result gives a complete answer under a square-exponential moment
assumption.

\begin{theorem}[Zero-SNR analyticity rigidity]
\label{thm:main}
Let $X$ be a real random variable satisfying
\begin{equation}
    \E e^{\beta X^2}<\infty
    \label{eq:subgaussian}
\end{equation}
for some $\beta>0$.  Then
\begin{equation}
    \mmse_X(s)\text{ is analytic at }s=0
    \quad\Longleftrightarrow\quad
    X\text{ is Gaussian},
    \label{eq:main-equivalence}
\end{equation}
where a constant random variable is regarded as a degenerate Gaussian.
Here analyticity means that the function defined for $s\geq0$ agrees on some
$[0,\varepsilon)$ with a holomorphic function on a complex disc centered at
zero.
\end{theorem}

For $X\sim N(\mu,\sigma^2)$, direct Gaussian conditioning gives
\begin{equation}
    \mmse_X(s)=\frac{\sigma^2}{1+\sigma^2s},
    \label{eq:gaussian-mmse}
\end{equation}
so the implication ``Gaussian $\Rightarrow$ analytic'' in
\eqref{eq:main-equivalence} is immediate.  The
content of the theorem is that no other sub-Gaussian input has an analytic
MMSE at the zero-SNR boundary.

\begin{corollary}[Rational MMSE rigidity]
\label{cor:rational}
Under \eqref{eq:subgaussian}, suppose there exists a rational function
$R\in\R(s)$ such that
\[
    \mmse_X(s)=R(s),\qquad s>0.
\]
Then $X$ is Gaussian and \eqref{eq:gaussian-mmse} holds.
\end{corollary}

\begin{corollary}[Mutual-information rigidity]
\label{cor:mutual-information}
Let $I_X(s)=I(X;\sqrt{s}X+Z)$.  Under \eqref{eq:subgaussian}, $I_X$ is
analytic at $s=0$ if and only if $X$ is Gaussian.  In that case
\begin{equation}
    I_X(s)=\frac12\log(1+\sigma^2s).
\end{equation}
\end{corollary}

The main contributions are the analyticity characterization in
Theorem~\ref{thm:main} and the connection between complex MGF zeros and
large-order coefficient growth.  Corollaries~\ref{cor:rational}
and~\ref{cor:mutual-information} are direct consequences of that
characterization.

The proof uses three ideas.  First, the square-exponential hypothesis is a
standard sub-Gaussian tail condition
\cite{BuldyginKozachenko2000,Vershynin2018}.  It makes the moment-generating
function $M(z)=\E e^{zX}$ entire of order at most two.  Hadamard
factorization implies that a zero-free $M$ must be Gaussian; hence every
non-Gaussian law in our class has a nonzero complex MGF zero.  This is in the
circle of ideas around the Marcinkiewicz theorem for entire characteristic
functions
\cite{Marcinkiewicz1939,Lukacs1970,EremenkoFryntov2021,Boas1954,Levin1996}.
Second, the
posterior mean in \eqref{eq:channel} can be written through the backward heat
evolution
\begin{equation}
    Q(s,z)=\E e^{zX-sX^2/2},
    \qquad Q_s=-\tfrac12 Q_{zz},
\end{equation}
and the captured posterior energy is a Gaussian diagonal integral of
$F=Q_z^2/Q$.  A simple zero of $Q$ is therefore a pole of the integrand;
higher-multiplicity zeros split into simple branches under the local heat
flow analyzed below.  Third, an exact double-Borel identity transfers the
resulting spatial poles to finite
singularities in the Borel plane.  The use of heat evolution to study complex
zeros has classical precedents in the de Bruijn--Newman theory and its later
developments \cite{deBruijn1950,Newman1976,GriffinEtAl2019}; related
connections between backward heat flow and Hermite-polynomial zero dynamics
appear in \cite{Kabluchko2024}.

The simple-zero mechanism is a moving-pole version of the identity
\begin{equation}
    \cB\!\left[-\frac r\zeta
      \sum_{n\geq0}(2n-1)!!\left(\frac{s}{\zeta^2}\right)^n\right](\xi)
    =-\frac r\zeta\left(1-\frac{2\xi}{\zeta^2}\right)^{-1/2}.
\end{equation}
Multiple zeros require a finer local analysis.  Backward heat flow splits a
zero of multiplicity $m$ along the roots of the probabilists' Hermite
polynomial $\He_m$, and a uniform saddle-point estimate shows that the Borel
coefficients contain the nonvanishing factor
$n^{-m/2}e^{r_m\sqrt{2n}}$, where $r_m$ is the largest root of $\He_m$.
We use standard facts about Hermite zeros and their asymptotics from
\cite{Szego1975,Dominici2007,IsmailLi1992,Deift1999} and the classical
framework of Borel summation and resurgence
\cite{Watson1912,Sokal1980,Balser1994,Sauzin2014,Dorigoni2019,Costin2008}.

The remainder of the paper is organized as follows.  Section~\ref{sec:setup}
normalizes the problem and proves the MGF zero lemma.
Section~\ref{sec:heat} derives the backward-heat representation and identifies
the actual zero-SNR asymptotic expansion.  Section~\ref{sec:borel} establishes
the double-Borel identity and finite-disc localization.  Simple and multiple
zeros are treated in Sections~\ref{sec:simple} and \ref{sec:multiple}.
Section~\ref{sec:continuation} proves local continuation and excludes
cancellation.  The theorem and its consequences are completed in
Section~\ref{sec:completion}.

Figure~\ref{fig:proof-dependency} summarizes the logical structure of the
proof.  Gray nodes are classical or standard ingredients, blue nodes are
problem-specific technical steps, orange nodes are the principal
propositions, and the green node is the main rigidity theorem.

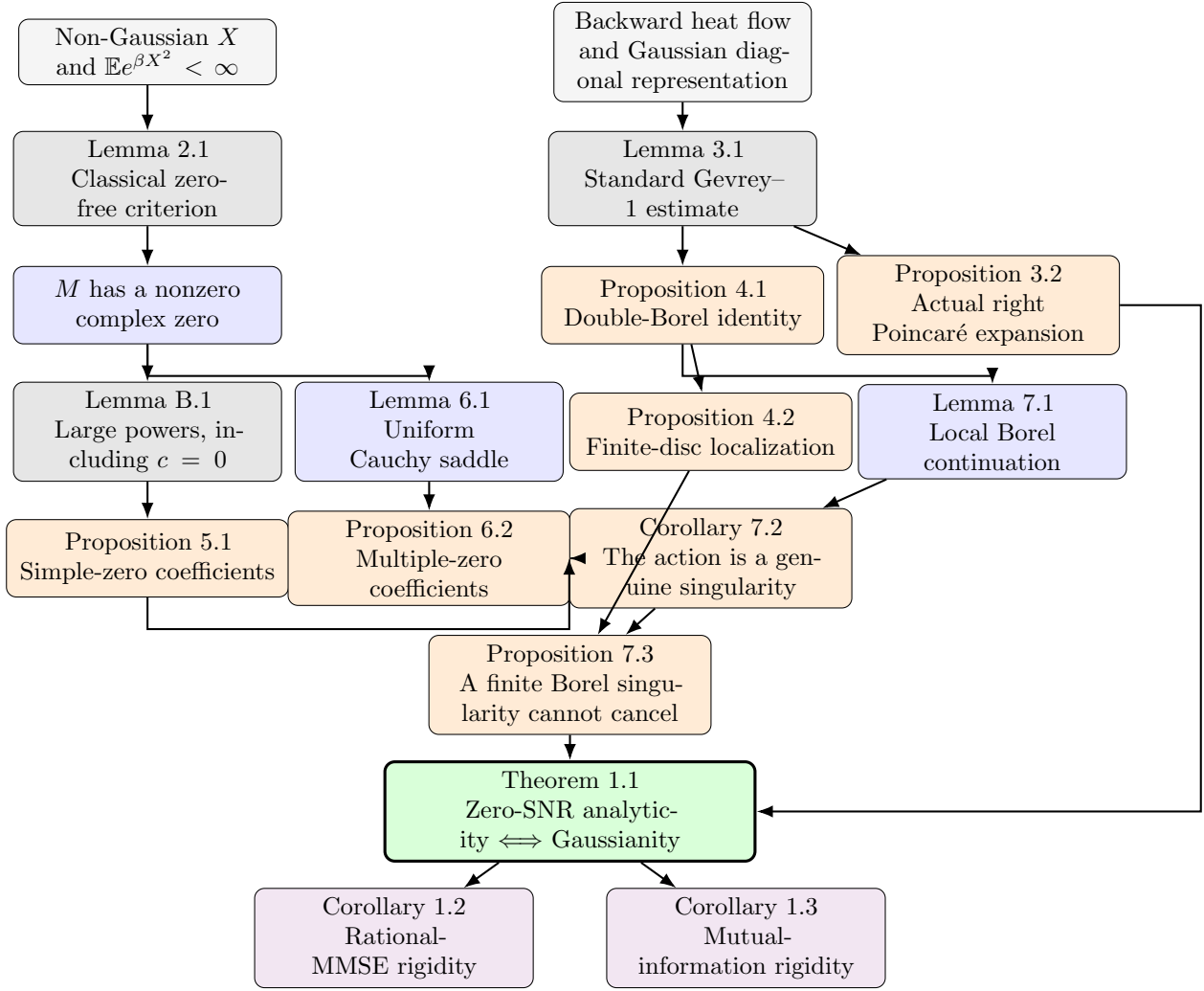
\begin{figure}[t]
\centering
\resizebox{\textwidth}{!}{%
\begin{tikzpicture}[
    >=Latex,
    node distance=8mm and 9mm,
    every node/.style={font=\small},
    input/.style={draw,rounded corners,fill=black!4,align=center,
      text width=3.4cm,minimum height=9mm},
    classical/.style={draw,rounded corners,fill=black!10,align=center,
      text width=3.55cm,minimum height=11mm},
    technical/.style={draw,rounded corners,fill=blue!10,align=center,
      text width=3.55cm,minimum height=11mm},
    proposition/.style={draw,rounded corners,fill=orange!16,align=center,
      text width=3.75cm,minimum height=11mm},
    conclusion/.style={draw,rounded corners,fill=green!15,very thick,
      align=center,text width=5.0cm,minimum height=12mm},
    corollary/.style={draw,rounded corners,fill=violet!10,align=center,
      text width=3.7cm,minimum height=10mm},
    arrow/.style={->,thick}
]

\node[input] (nong) at (-5.8,0)
  {Non-Gaussian $X$ and $\E e^{\beta X^2}<\infty$};
\node[input] (heat) at (1.8,0)
  {Backward heat flow and Gaussian diagonal representation};

\node[classical] (zero) at (-5.8,-1.8)
  {Lemma~\ref{lem:zero-free}\\Classical zero-free criterion};
\node[classical] (gevrey) at (1.8,-1.8)
  {Lemma~\ref{lem:gevrey}\\Standard Gevrey--1 estimate};

\node[technical] (mgfzero) at (-5.8,-3.6)
  {$M$ has a nonzero complex zero};
\node[proposition] (borel) at (1.8,-3.6)
  {Proposition~\ref{prop:double-borel}\\Double-Borel identity};
\node[proposition] (poincare) at (6.0,-3.6)
  {Proposition~\ref{prop:poincare}\\Actual right Poincar\'e expansion};

\node[classical] (large) at (-5.8,-5.4)
  {Lemma~\ref{lem:analytic-large-powers}\\Large powers, including $c=0$};
\node[technical] (uniform) at (-1.8,-5.4)
  {Lemma~\ref{lem:uniform-extraction}\\Uniform Cauchy saddle};
\node[proposition] (localize) at (2.2,-5.4)
  {Proposition~\ref{prop:localization}\\Finite-disc localization};
\node[technical] (continue) at (6.2,-5.4)
  {Lemma~\ref{lem:heat-pole}\\Local Borel continuation};

\node[proposition] (simple) at (-5.8,-7.2)
  {Proposition~\ref{prop:simple-asymptotics}\\Simple-zero coefficients};
\node[proposition] (multiple) at (-1.8,-7.2)
  {Proposition~\ref{prop:multiple-asymptotics}\\Multiple-zero coefficients};
\node[proposition] (genuine) at (2.2,-7.2)
  {Corollary~\ref{cor:genuine-action}\\The action is a genuine singularity};

\node[proposition] (nocancel) at (0.2,-9.0)
  {Proposition~\ref{prop:no-cancellation}\\A finite Borel singularity cannot cancel};

\node[conclusion] (main) at (0.2,-10.8)
  {Theorem~\ref{thm:main}\\Zero-SNR analyticity $\Longleftrightarrow$ Gaussianity};

\node[corollary] (rational) at (-2.3,-12.6)
  {Corollary~\ref{cor:rational}\\Rational-MMSE rigidity};
\node[corollary] (mutual) at (2.7,-12.6)
  {Corollary~\ref{cor:mutual-information}\\Mutual-information rigidity};

\draw[arrow] (nong) -- (zero);
\draw[arrow] (zero) -- (mgfzero);
\draw[arrow] (heat) -- (gevrey);
\draw[arrow] (gevrey) -- (poincare);
\draw[arrow] (gevrey) -- (borel);
\draw[arrow] (mgfzero) -- (large);
\draw[arrow] (mgfzero.south) -- ++(0,-0.45) -| (uniform.north);
\draw[arrow] (large) -- (simple);
\draw[arrow] (uniform) -- (multiple);
\draw[arrow] (borel) -- (localize);
\draw[arrow] (borel.south) -- ++(0,-0.45) -| (continue.north);
\draw[arrow] (simple.south) -- ++(0,-0.45) -| (genuine.west);
\draw[arrow] (multiple) -- (genuine);
\draw[arrow] (continue) -- (genuine);
\draw[arrow] (genuine) -- (nocancel);
\draw[arrow] (localize) -- (nocancel);
\draw[arrow] (poincare.east) -- ++(1.15,0) |- (main.east);
\draw[arrow] (nocancel) -- (main);
\draw[arrow] (main) -- (rational);
\draw[arrow] (main) -- (mutual);

\end{tikzpicture}%
}
\caption{Dependency graph for the proof of the zero-SNR analyticity
rigidity theorem.  The graph distinguishes standard inputs from the
problem-specific Borel and zero-cluster analysis.}
\label{fig:proof-dependency}
\end{figure}

\paragraph{Notation.}
The Euclidean modulus on $\C$ is denoted by $|\cdot|$.  If
$G(s)=\sum_{n\geq0}g_ns^n$ is formal, then $[s^n]G=g_n$ and its ordinary
Borel transform is $\cB G(\xi)=\sum_{n\geq0}g_n\xi^n/n!$.  We use the
convention $(-1)!!=1$ and write
\begin{equation}
    d_n=(2n-1)!!,
    \qquad d_0=1.
\end{equation}
For a zero cluster coalescing at $z_0$, its diagonal series and Borel
transform are denoted by $\widehat A_{z_0}$ and
$B_{z_0}=\cB\widehat A_{z_0}$, respectively.  All contours are positively
oriented unless stated otherwise.

\section{Normalization and zeros of the MGF}
\label{sec:setup}

Translation and scaling of the input give
\begin{equation}
    \mmse_{X+c}(s)=\mmse_X(s),
    \qquad
    \mmse_{aX}(s)=a^2\mmse_X(a^2s).
    \label{eq:affine-mmse}
\end{equation}
Indeed, the observation generated by $X+c$ is
$\sqrt{s}X+Z+\sqrt{s}c$; subtracting the deterministic last term does not
change the generated sigma-field, while
\begin{equation}
 (X+c)-\E[X+c\mid \sqrt{s}(X+c)+Z]
 =X-\E[X\mid\sqrt{s}X+Z].
\end{equation}
For scaling, the channel for $aX$ at SNR $s$ is the channel for $X$ at SNR
$a^2s$, up to replacing the observation by its negative when $a<0$, and its
conditional estimation error is multiplied by $a$.  Squaring and taking
expectations gives the second identity in \eqref{eq:affine-mmse}.
The square-exponential moment condition is preserved, with a possibly smaller
exponent, under both operations.  Indeed,
$(x-c)^2\leq2x^2+2c^2$, and scaling merely replaces the admissible exponent
by its quotient with $a^2$.  Thus the normalization does not leave the class
of input laws covered by the theorem.
Thus, apart from the constant case, we may and do normalize
\begin{equation}
    \E X=0,
    \qquad \E X^2=1.
    \label{eq:normalization}
\end{equation}

Let
\begin{equation}
    M(z)=\E e^{zX},\qquad z\in\C.
    \label{eq:mgf}
\end{equation}
By Young's inequality,
\begin{equation}
    |z||x|\leq \frac{\beta x^2}{2}+\frac{|z|^2}{2\beta},
\end{equation}
and hence
\begin{equation}
    |M(z)|
    \leq \E e^{|z||X|}
    \leq \E e^{\beta X^2/2}\,e^{|z|^2/(2\beta)}.
    \label{eq:mgf-growth}
\end{equation}
Consequently $M$ is entire of order at most two and finite type.

The following zero-free criterion is classical.  It is an immediate
consequence of Hadamard factorization and is closely related to the
Marcinkiewicz characterization of Gaussian laws
\cite{Marcinkiewicz1939,EremenkoFryntov2021,Boas1954,Levin1996}.  We include
the short proof only to make the argument self-contained; no originality is
claimed for this lemma.

\begin{lemma}[Classical zero-free criterion]
\label{lem:zero-free}
If the entire function $M$ in \eqref{eq:mgf} has no zero in $\C$, then $X$
is Gaussian.
\end{lemma}

\begin{proof}
Hadamard factorization for a zero-free entire function of finite order at most
two gives
\begin{equation}
    M(z)=e^{P(z)},
\end{equation}
where $P$ is a polynomial of degree at most two.  Since $M(0)=1$, we choose
$P(0)=0$.  For real $t$, $M(t)>0$, so
$\operatorname{Im}P(t)\in2\pi\mathbb Z$.  This integer-valued continuous
function of $t$ vanishes at $t=0$ and hence vanishes identically.  Thus
$P(t)$ is real for every real $t$, and the identity theorem applied to
$P(z)-\overline{P(\overline z)}$ shows that all coefficients of $P$ are
real.  Therefore
\begin{equation}
    P(z)=az+bz^2,\qquad a,b\in\R.
\end{equation}
Moreover,
\begin{equation}
    a=P'(0)=\E X,
    \qquad
    2b=P''(0)=\Var(X)\geq0.
\end{equation}
Thus
$M(z)=\exp(\mu z+\sigma^2z^2/2)$, the MGF of
$N(\mu,\sigma^2)$.
\end{proof}

In particular, a non-Gaussian input satisfying \eqref{eq:subgaussian} has at
least one nonzero complex MGF zero.  The proof below shows that such zeros are
exactly the obstruction to zero-SNR analyticity.

\section{Backward heat flow and the zero-SNR expansion}
\label{sec:heat}

\subsection{A heat representation of posterior energy}

Define
\begin{equation}
    Q(s,z)=\E\exp\!\left(zX-\frac{sX^2}{2}\right).
    \label{eq:Q-def}
\end{equation}
Condition \eqref{eq:subgaussian} permits differentiation under the
expectation for $|s|<\beta$, locally uniformly in $z$.  Therefore $Q$ is
jointly holomorphic there, entire in $z$, and satisfies
\begin{equation}
    Q_s=-\frac12Q_{zz},
    \qquad Q(0,z)=M(z).
    \label{eq:backward-heat}
\end{equation}
Formally, $Q(s,\cdot)=e^{-(s/2)\partial_z^2}M$.

Let $\phi(y)=(2\pi)^{-1/2}e^{-y^2/2}$.  The output density in
\eqref{eq:channel} is
\begin{equation}
    p_s(y)=\phi(y)Q(s,\sqrt{s}\,y),
    \label{eq:output-density}
\end{equation}
because
\begin{align}
 p_s(y)
 &=\E\phi(y-\sqrt{s}X) \\
 &=\phi(y)\E\exp\left(\sqrt{s}\,yX-\frac{sX^2}{2}\right)
 =\phi(y)Q(s,\sqrt{s}\,y).
 \label{eq:output-density-derivation}
\end{align}
Likewise, Bayes' formula gives the conditional first moment as the ratio of
the corresponding unnormalized integrals:
\begin{equation}
    \E[X\mid Y_s=y]
    =\frac{Q_z(s,\sqrt{s}\,y)}{Q(s,\sqrt{s}\,y)}.
    \label{eq:posterior-mean}
\end{equation}
Indeed, the numerator after extracting $\phi(y)$ is
$\E[Xe^{\sqrt{s}yX-sX^2/2}]=Q_z(s,\sqrt{s}y)$, and the denominator is
$Q(s,\sqrt{s}y)$.
Set
\begin{equation}
    A(s)=\E\!\left[\E[X\mid Y_s]^2\right]
    =\E X^2-\mmse_X(s)
    \label{eq:A-def}
\end{equation}
and
\begin{equation}
    F(s,z)=\frac{Q_z(s,z)^2}{Q(s,z)}.
    \label{eq:F-def}
\end{equation}
Using \eqref{eq:output-density}--\eqref{eq:posterior-mean}, we first obtain
\begin{align}
 A(s)
 &=\int_{\R}
   \left(\frac{Q_z(s,\sqrt{s}y)}{Q(s,\sqrt{s}y)}\right)^2
   \phi(y)Q(s,\sqrt{s}y)\dd y \\
 &=\int_{\R}\phi(y)
   \frac{Q_z(s,\sqrt{s}y)^2}{Q(s,\sqrt{s}y)}\dd y.
 \label{eq:posterior-energy-before-change}
\end{align}
Changing variables $z=\sqrt{s}\,y$ now gives the exact diagonal integral
\begin{equation}
    \boxed{
    A(s)=\frac{1}{\sqrt{2\pi s}}
    \int_{\R}e^{-z^2/(2s)}F(s,z)\dd z.}
    \label{eq:diagonal-integral}
\end{equation}
Because $Q(0,0)=1$, the function $F$ is holomorphic in a bidisc centered at
$(0,0)$.  Away from that bidisc, the moving zeros of $Q(s,\cdot)$ are the
poles responsible for divergence.

\subsection{The actual right asymptotic series}

Expand near $(0,0)$ as
\begin{equation}
    F(s,z)=\sum_{k,\ell\geq0}f_{k,\ell}s^kz^\ell.
    \label{eq:F-series}
\end{equation}
The centered Gaussian moments satisfy
\begin{equation}
    \frac{1}{\sqrt{2\pi s}}
    \int_{\R}e^{-z^2/(2s)}z^{2q}\dd z
    =(2q-1)!!\,s^q,
    \label{eq:gaussian-moments}
\end{equation}
whereas odd powers integrate to zero.  This suggests the formal series
\begin{equation}
    \widehat A(s)
    =\sum_{k,q\geq0}f_{k,2q}(2q-1)!!\,s^{k+q}.
    \label{eq:Ahat}
\end{equation}

The next lemma is a standard consequence of bidisc Cauchy estimates and the
growth of Gaussian moments.  It is recorded for later use and is not claimed
as a new general Gevrey theorem; compare the standard treatments of Gevrey
asymptotics and Borel summation in \cite{Balser1994,Costin2008}.

\begin{lemma}[Standard Gevrey--1 estimate]
\label{lem:gevrey}
There are constants $C_0,C_1>0$ such that
\begin{equation}
    \bigl|[s^n]\widehat A(s)\bigr|
    \leq C_0C_1^n n!,
    \qquad n\geq0.
    \label{eq:gevrey-bound}
\end{equation}
Thus the formal zero-SNR expansion is Gevrey of order at most one.
\end{lemma}

\begin{proof}
Holomorphy of $F$ in a bidisc gives
$|f_{k,2q}|\leq C\delta^{-k}\rho^{-2q}$ by Cauchy's estimate.  Since
$(2q-1)!!\leq2^q q!\leq2^q n!$ when $k+q=n$, summing the $n+1$ diagonal
terms and absorbing the polynomial factor into an exponential proves
\eqref{eq:gevrey-bound}.
\end{proof}

\begin{proposition}[Poincar\'e expansion at zero]
\label{prop:poincare}
The series \eqref{eq:Ahat} is the right Poincar\'e asymptotic expansion of
$A(s)$: for every $N\geq0$,
\begin{equation}
    A(s)-\sum_{n=0}^{N}[s^n]\widehat A(s)s^n
    =O(s^{N+1}),
    \qquad s\downarrow0.
    \label{eq:poincare}
\end{equation}
Consequently, if $A$ is analytic at zero, then \eqref{eq:Ahat} has positive
ordinary radius of convergence.
\end{proposition}

\begin{proof}
For real $s>0$ and $z\in\R$, Cauchy--Schwarz under the positive tilted
measure in \eqref{eq:Q-def} yields
\begin{equation}
    Q_z(s,z)^2\leq Q(s,z)Q_{zz}(s,z),
    \qquad 0\leq F(s,z)\leq Q_{zz}(s,z).
    \label{eq:F-domination}
\end{equation}
Completing the square gives
\begin{equation}
    \frac{e^{-z^2/(2s)}}{\sqrt{2\pi s}}Q_{zz}(s,z)
    =\E\!\left[
      X^2\frac{e^{-(z-sX)^2/(2s)}}{\sqrt{2\pi s}}
    \right].
    \label{eq:complete-square}
\end{equation}
Fix a small $\delta>0$ inside the spatial holomorphy disc.  The integral of
the right-hand side of \eqref{eq:complete-square} over $|z|>\delta$ equals
\begin{equation}
    \E\!\left[X^2
      \Pp\{|sX+\sqrt{s}G|>\delta\mid X\}\right],
    \label{eq:outer-probability}
\end{equation}
where $G\sim N(0,1)$.  The event is contained in
\begin{equation}
    \{|X|>\delta/(2s)\}
    \cup
    \{|G|>\delta/(2\sqrt{s})\}.
\end{equation}
The square-exponential moment of $X$ and the Gaussian tail bound imply that
\eqref{eq:outer-probability} is $O(e^{-c/s})$ for some $c>0$.  Hence the
outer part of \eqref{eq:diagonal-integral} is flat at zero.

On $|z|<\delta$, apply the two-variable Taylor formula to $F$, assigning
parabolic degree $k+\ell/2$ to $s^kz^\ell$.  Termwise integration uses
\eqref{eq:gaussian-moments}; the integrated Taylor remainder is
$O(s^{N+1})$.  Combining the inner and outer estimates proves
\eqref{eq:poincare}.  Uniqueness of the Taylor expansion of a holomorphic
extension proves the final assertion.  A coefficientwise version of the
Taylor-remainder estimate is given in Appendix~\ref{app:poincare-details}.
\end{proof}

\section{The Borel transform and finite-disc localization}
\label{sec:borel}

\subsection{An exact double-Borel identity}

Write
\begin{equation}
    F(s,z)=\sum_{k\geq0}F_k(z)s^k
\end{equation}
and take the exponential Borel transform only in the $s$ index:
\begin{equation}
    \widetilde F(u,z)
    =\sum_{k\geq0}F_k(z)\frac{u^k}{k!}.
    \label{eq:partial-borel}
\end{equation}
For a holomorphic function $G$, define its circular heat-Borel average by
\begin{equation}
    (\cC G)(v)
    =\frac{1}{2\pi}\int_0^{2\pi}
      G(\sqrt{2v}\cos\theta)\dd\theta.
    \label{eq:circular-average}
\end{equation}
The value is understood through its power series near $v=0$, and thereafter
by analytic continuation.

\begin{proposition}[Double-Borel identity]
\label{prop:double-borel}
The ordinary Borel transform of \eqref{eq:Ahat} obeys the formal identity
\begin{equation}
    \boxed{
    \cB\widehat A(\xi)
    =\frac{\dd}{\dd\xi}\left[
      \xi\int_0^1
      \cC\widetilde F(t\xi,\cdot)((1-t)\xi)\dd t
    \right].}
    \label{eq:double-borel}
\end{equation}
It is an analytic identity wherever both sides converge near the origin.
\end{proposition}

\begin{proof}
The circular moments are
\begin{equation}
    \frac{1}{2\pi}\int_0^{2\pi}\cos^{2q}\theta\dd\theta
    =\frac{(2q)!}{2^{2q}(q!)^2}.
\end{equation}
Therefore
\begin{equation}
    \cC\widetilde F(u,\cdot)(v)
    =\sum_{k,q\geq0}
      f_{k,2q}(2q-1)!!\frac{u^k}{k!}\frac{v^q}{q!}.
    \label{eq:circular-series}
\end{equation}
Using
\begin{equation}
    \int_0^1t^k(1-t)^q\dd t
    =\frac{k!q!}{(k+q+1)!},
\end{equation}
the coefficient of $f_{k,2q}$ inside the brackets of
\eqref{eq:double-borel} is
\begin{equation}
    \frac{(2q-1)!!}{(k+q+1)!}\xi^{k+q+1}.
\end{equation}
Differentiation gives precisely the coefficient prescribed by the Borel
transform of \eqref{eq:Ahat}.
\end{proof}

\subsection{Localization to finitely many zero clusters}

Fix $R>0$ such that $M$ has no zero on $|z|=R$.  Since
$Q(s,\cdot)\to M$ uniformly on compact sets, Rouch\'e's theorem implies that,
for small $|s|$, the number of zeros of $Q(s,\cdot)$ in $|z|<R$, counted
with multiplicity, is constant.  Surround the resulting finitely many zero
clusters by disjoint circles $C_j$.  For $z$ outside $C_j$, define
\begin{equation}
    P_j(s,z)=\frac{1}{2\pi i}\int_{C_j}
      \frac{F(s,w)}{z-w}\dd w.
    \label{eq:cluster-principal-part}
\end{equation}
By residues, $P_j$ is the sum of the principal parts of $F$ at the zeros in
that cluster.  The contour representation also shows that the cluster sum
is holomorphic in $s$, even if the individual roots are only holomorphic in
$\sqrt{s}$.

\begin{proposition}[Finite-disc localization]
\label{prop:localization}
Let
\begin{equation}
    H(s,z)=F(s,z)-\sum_jP_j(s,z).
    \label{eq:H-def}
\end{equation}
After decreasing the radii slightly, $H$ is holomorphic on a bidisc
$|s|<\delta$, $|z|<R$.  The Borel transform of the diagonal series generated
by $H$ is holomorphic for $|\xi|<R^2/2$.
\end{proposition}

\begin{proof}
For each fixed small $s$, the residue theorem identifies $P_j(s,\cdot)$
with the complete sum of the principal parts of $F(s,\cdot)$ at the poles
inside $C_j$.  The removability is joint in $(s,z)$, as can be seen without
choosing individual zero branches.  If $z$ lies inside $C_j$ and is
temporarily kept away from the poles, Cauchy's residue theorem gives
\begin{equation}
 F(s,z)-P_j(s,z)
 =\frac{1}{2\pi i}\int_{C_j}\frac{F(s,w)}{w-z}\dd w.
 \label{eq:cluster-holomorphic-remainder}
\end{equation}
The right-hand side is jointly holomorphic in $(s,z)$ throughout the
interior of $C_j$, because $F$ is jointly holomorphic on the fixed contour.
It therefore supplies the joint removable extension across every moving
pole in that cluster.  The terms $P_\ell$ for $\ell\neq j$ are holomorphic
near $C_j$, so the same conclusion holds for $H$.  There are only finitely
many clusters in $|z|<R$; decreasing $\delta$ if necessary gives joint
holomorphy on $|s|<\delta$, $|z|<R$.

Write $H(s,z)=\sum h_{k,\ell}s^kz^\ell$.  Cauchy estimates give
the desired bounds only on compactly contained bidiscs, so we keep the
auxiliary radii explicit.  Fix
\begin{equation}
    0<\delta_0<\delta,
    \qquad 0<\rho<R.
\end{equation}
Since $H$ is holomorphic on a neighborhood of the compact torus
$|s|=\delta_0$, $|z|=\rho$, there is a constant
$C_{\delta_0,\rho}$ such that
\begin{equation}
    |h_{k,2q}|\leq
    C_{\delta_0,\rho}\delta_0^{-k}\rho^{-2q}.
    \label{eq:H-cauchy}
\end{equation}
Let $b_n^H$ be the coefficient of $\xi^n$ in the Borel transform of the
diagonal series generated by $H$.  Since
\begin{equation}
    (2q-1)!!\leq 2^q q!,
    \qquad
    \frac{q!}{(q+k)!}\leq\frac1{k!},
\end{equation}
we obtain
\begin{align}
    |b_n^H|
    &\leq C_{\delta_0,\rho}
      \left(\frac{2}{\rho^2}\right)^n
      \sum_{k=0}^n
      \frac{(\rho^2/(2\delta_0))^k}{k!} \\
    &\leq C'_{\delta_0,\rho}
      \left(\frac{2}{\rho^2}\right)^n.
    \label{eq:H-borel-bound}
\end{align}
Thus the Borel transform is holomorphic for $|\xi|<\rho^2/2$.
Because $\rho<R$ was arbitrary, these discs exhaust
$|\xi|<R^2/2$, proving the assertion.  Notice that no estimate on the
open boundary $|z|=R$ has been used.
\end{proof}

Proposition~\ref{prop:localization} is the first noncancellation mechanism:
zeros outside a fixed finite disc cannot remove a singularity generated by a
zero strictly inside it.

\section{The singularity generated by a simple zero}
\label{sec:simple}

Suppose $M(z_0)=0$ and $M'(z_0)\neq0$.  The implicit-function theorem gives
a holomorphic zero $\zeta(s)$ of $Q$ with $\zeta(0)=z_0$.  From
\eqref{eq:backward-heat},
\begin{equation}
    \zeta'(0)
    =-\frac{Q_s(0,z_0)}{Q_z(0,z_0)}
    =\frac{M''(z_0)}{2M'(z_0)}.
    \label{eq:zeta-prime}
\end{equation}
At $z=\zeta(s)$, the principal part of $F$ is
\begin{equation}
    \frac{r(s)}{z-\zeta(s)},
    \qquad
    r(s)=Q_z(s,\zeta(s)),\quad r(0)=M'(z_0)\neq0.
    \label{eq:simple-principal-part}
\end{equation}
Indeed, writing $Q(s,z)=(z-\zeta(s))q(s,z)$ gives
$q(s,\zeta(s))=Q_z(s,\zeta(s))=r(s)$ and
\begin{equation}
 \frac{Q_z(s,z)^2}{Q(s,z)}
 =\frac{r(s)}{z-\zeta(s)}+O(1)
 \qquad (z\to\zeta(s)).
\end{equation}
Its contribution to the diagonal series is
\begin{equation}
    \widehat A_{z_0}(s)
    =-\sum_{q\geq0}(2q-1)!!\,
      r(s)\zeta(s)^{-2q-1}s^q.
    \label{eq:simple-diagonal}
\end{equation}

\begin{proposition}[Simple-zero coefficient asymptotics]
\label{prop:simple-asymptotics}
Let $d_n=(2n-1)!!$.  Then
\begin{equation}
    [s^n]\widehat A_{z_0}(s)
    \sim
    -\frac{M'(z_0)}{z_0}
    \exp\!\left[-\frac{z_0M''(z_0)}{2M'(z_0)}\right]
    d_nz_0^{-2n}.
    \label{eq:simple-asymptotic}
\end{equation}
The prefactor is nonzero.  Consequently the cluster Borel transform has
radius $|z_0|^2/2$ in the action variable.  Section~\ref{sec:continuation}
will identify the corresponding action
\begin{equation}
    \xi_0=\frac{z_0^2}{2}.
    \label{eq:simple-action}
\end{equation}
as a nonremovable singularity.
\end{proposition}

\begin{proof}
Put
\begin{equation}
    a(s)=\frac{r(s)}{r(0)},
    \qquad
    h(s)=\log\frac{\zeta(s)}{z_0},
    \qquad
    c=h'(0)=\frac{\zeta'(0)}{z_0}.
\end{equation}
In the coefficient of $s^n$ in \eqref{eq:simple-diagonal}, put $q=n-k$.
After division by $-r(0)d_nz_0^{-2n-1}$, the $k$th summand is
\begin{equation}
    T_{n,k}=\frac{d_{n-k}}{d_n}z_0^{2k}
    [s^k]a(s)e^{-N_{n,k}h(s)},
    \qquad N_{n,k}=2(n-k)+1.
    \label{eq:simple-normalized-summand}
\end{equation}
Lemma~\ref{lem:analytic-large-powers} gives, for every fixed $k$,
\begin{equation}
    T_{n,k}\longrightarrow\frac{(-cz_0^2)^k}{k!},
    \label{eq:simple-fixed-k-limit}
\end{equation}
including the zero-drift case $c=0$.  Its majorant and tail estimates justify
dominated convergence for this triangular array.  Therefore
\begin{equation}
    \sum_{k=0}^nT_{n,k}\longrightarrow e^{-cz_0^2}.
\end{equation}
Using $r(0)=M'(z_0)$ and \eqref{eq:zeta-prime} yields
\eqref{eq:simple-asymptotic}.

Finally,
\begin{equation}
    \frac{d_n}{n!}
    =\frac{(2n)!}{2^n(n!)^2}
    \sim\frac{2^n}{\sqrt{\pi n}}.
    \label{eq:double-factorial-central-binomial}
\end{equation}
After the rescaling $\xi=(z_0^2/2)x$, the Borel coefficients are a nonzero
constant times $n^{-1/2}(1+o(1))$, proving the asserted radius.
\end{proof}

\begin{remark}[The binary zero-drift check]
For $M(z)=\cosh z$ one has $Q(s,z)=e^{-s/2}\cosh z$ and
$\zeta(s)\equiv z_0$ at every zero of $\cosh$, so $c=0$.  The additive
estimate in Lemma~\ref{lem:analytic-large-powers} still applies.  A
multiplicative error relative to $(-2nc)^k/k!$ would not be meaningful in
this case.
\end{remark}

\begin{example}[Symmetric binary input]
\label{ex:rademacher}
If $\Pp\{X=1\}=\Pp\{X=-1\}=1/2$, then $M(z)=\cosh z$.  Its nearest zeros
are $z_0=\pm i\pi/2$, so the nearest Borel action is
\begin{equation}
    \xi_0=-\frac{\pi^2}{8}.
\end{equation}
Thus the zero-SNR series is alternating at leading order, Gevrey--1, and
divergent.  This recovers the classical nonanalyticity of the binary-input
MMSE at zero and identifies its controlling complex action.
\end{example}

\section{Multiple zeros and Hermite splitting}
\label{sec:multiple}

Suppose $z_0$ is a zero of $M$ of multiplicity $m\geq2$.  Write locally
\begin{equation}
    M(z_0+w)=aw^m(1+bw+O(w^2)),
    \qquad a\neq0.
    \label{eq:multiple-germ}
\end{equation}
Set $s=\tau^2$ and $w=\tau x$.  Applying the backward heat operator to the
Taylor germ yields
\begin{equation}
    Q(\tau^2,z_0+\tau x)
    =a\tau^m\left[
      \He_m(x)+b\tau\He_{m+1}(x)+O(\tau^2)
    \right],
    \label{eq:hermite-normal-form}
\end{equation}
To verify \eqref{eq:hermite-normal-form}, use the polynomial identity
\begin{align}
 e^{-(\tau^2/2)\partial_w^2}w^\ell
 &=\sum_{j=0}^{\lfloor\ell/2\rfloor}
   \frac{(-1)^j}{2^jj!}\frac{\ell!}{(\ell-2j)!}
   \tau^{2j}w^{\ell-2j} \\
 &=\tau^\ell\He_\ell(w/\tau).
 \label{eq:heat-hermite-identity}
\end{align}
For related connections between backward heat flow and Hermite-polynomial
zero dynamics, see \cite{Kabluchko2024}.
The terms $aw^m$ and $abw^{m+1}$ in
\eqref{eq:multiple-germ} therefore become
$a\tau^m\He_m(x)$ and $ab\tau^{m+1}\He_{m+1}(x)$, respectively.  Every
remaining Taylor monomial has degree at least $m+2$ and contributes
$O(\tau^{m+2})$ locally uniformly for bounded $x$, which proves the stated
normal form.
Here $\He_m$ is the probabilists' Hermite polynomial.  Its roots
\begin{equation}
    r_1<\cdots<r_m
\end{equation}
are real and simple.  The implicit-function theorem in $(\tau,x)$ gives the
split zeros
\begin{equation}
    \zeta_j(s)=z_0+r_j\sqrt{s}+bs+O(s^{3/2}).
    \label{eq:split-zeros}
\end{equation}
For the coefficient calculation, write
$x_j(\tau)=r_j+c_j\tau+O(\tau^2)$ and substitute into the bracket in
\eqref{eq:hermite-normal-form}.  The coefficient of $\tau$ is
\begin{equation}
 c_j\He_m'(r_j)+b\He_{m+1}(r_j).
\end{equation}
The Hermite identities
$\He_{m+1}(x)=x\He_m(x)-m\He_{m-1}(x)$ and
$\He_m'(x)=m\He_{m-1}(x)$ imply
$\He_{m+1}(r_j)=-\He_m'(r_j)$.  Since the roots are simple,
$\He_m'(r_j)\neq0$, and the vanishing of the displayed coefficient gives
$c_j=b$.  This proves \eqref{eq:split-zeros}.

The corresponding residues of $F$ are
\begin{equation}
    R_j(s)=Q_z(s,\zeta_j(s))
    =a\He_m'(r_j)s^{(m-1)/2}(1+O(\sqrt{s})).
    \label{eq:split-residues}
\end{equation}
Indeed, a simple zero of $Q$ at $z=\zeta_j(s)$ admits a factorization
\begin{equation}
 Q(s,z)=(z-\zeta_j(s))q_j(s,z),
 \qquad q_j(s,\zeta_j(s))=Q_z(s,\zeta_j(s)),
\end{equation}
and hence
\begin{equation}
 \frac{Q_z(s,z)^2}{Q(s,z)}
 =\frac{Q_z(s,\zeta_j(s))}{z-\zeta_j(s)}+O(1).
 \label{eq:residue-factorization}
\end{equation}
Differentiating \eqref{eq:hermite-normal-form} with respect to
$z=z_0+\tau x$ contributes the factor $\tau^{-1}$ and gives
\eqref{eq:split-residues}.
The individual branches use $\sqrt{s}$, but their cluster sum is
single-valued and holomorphic in $s$ by \eqref{eq:cluster-principal-part}.

\subsection{The homogeneous normal form}

For the leading germ $M(z_0+w)=aw^m$,
\begin{equation}
    Q=a s^{m/2}\He_m(w/\sqrt{s})
\end{equation}
and
\begin{equation}
    F
    =a s^{(m-2)/2}
      \frac{\He_m'(w/\sqrt{s})^2}{\He_m(w/\sqrt{s})}
    =a\sum_{k\geq0}c_{m,k}s^kw^{m-2-2k}.
    \label{eq:homogeneous-F}
\end{equation}
Partial fractions give
\begin{equation}
    \frac{\He_m'(x)^2}{\He_m(x)}
    =P_m(x)+\sum_{j=1}^m\frac{\He_m'(r_j)}{x-r_j},
    \label{eq:hermite-partial-fractions}
\end{equation}
so for all sufficiently large $k$,
\begin{equation}
    c_{m,k}=\sum_{j=1}^m
      \He_m'(r_j)r_j^{2k-m+1}.
    \label{eq:cmk}
\end{equation}
Let $r_m>0$ be the largest Hermite root.  Parity makes the contributions of
$r_m$ and $-r_m$ equal, and the remaining roots have smaller modulus.  Hence
\begin{equation}
    c_{m,k}\sim
    2\He_m'(r_m)r_m^{2k-m+1}>0.
    \label{eq:cmk-asymptotic}
\end{equation}
This already predicts the scale of the later saddle.  More explicitly, in
the homogeneous model the positive extreme branch has
$\alpha(\tau)=a\He_m'(r_m)\tau^{m-1}$ and
$\beta(\tau)=z_0+r_m\tau$.  With $p=n-k$ and
$L=2k-m+1\geq0$, the absolute value of its $k$th contribution, apart from a
$k$-independent factor, is
\begin{equation}
 S_{n,k}=d_p r_m^L\frac{(2p+1)_L}{L!},
 \label{eq:homogeneous-summand}
\end{equation}
where $(N)_L=N(N+1)\cdots(N+L-1)$ and $(N)_0=1$.  This follows from
$[\sigma^L](1+r_m\sigma)^{-N}=(-r_m)^L(N)_L/L!$.
For $p\geq1$, cancellation of the rising factorials gives the exact ratio
\begin{align}
 \frac{S_{n,k+1}}{S_{n,k}}
 &=\frac{2p\,r_m^2}{(L+1)(L+2)} \\
 &=\frac{r_m^2n}{2k^2}
   \left(1+O(n^{-1/2})\right),
 \label{eq:homogeneous-ratio}
\end{align}
where the second line is uniform when $k/\sqrt n$ ranges over a compact
subset of $(0,\infty)$.  The exact ratio is strictly decreasing in $k$
on the nonzero range of summands.  Its crossing of $1$ therefore locates
the saddle at
\begin{equation}
    k_*=r_m\sqrt{n/2}+O(1).
    \label{eq:multiple-saddle}
\end{equation}
The rising factorial cannot be replaced by $N^L$ when computing the leading
constant: for $L=O(\sqrt N)$,
\[
 \log\frac{(N)_L}{N^L}
 =\frac{L(L-1)}{2N}+O\!\left(\frac{L^3}{N^2}\right).
\]
At the saddle this correction tends to $r_m^2/2$.  It is retained below by
the quadratic term in Lemma~\ref{lem:uniform-extraction}; in the homogeneous
case that lemma has $B=0$ and $D=-r_m^2/2$.

\subsection{A uniform large-powers estimate}

The next lemma controls perturbations of the homogeneous germ.  It is the
main local asymptotic input for multiple zeros.  Its saddle-point method is
classical \cite{deBruijn1981,Olver1997,Wong2001}; the precise uniform
formulation below is tailored to the Hermite-split zero cluster considered
in this paper.

\begin{lemma}[Uniform Cauchy saddle]
\label{lem:uniform-extraction}
Let
\begin{equation}
    g(\sigma)=1+r\sigma+B\sigma^2+O(\sigma^3),
    \qquad r>0,
\end{equation}
be holomorphic and nonzero on $|\sigma|<\rho_0$.  Let $N\to\infty$ and let
$L=L_N$ be integers satisfying
\begin{equation}
    c_0\sqrt N\leq L\leq C_0\sqrt N
\end{equation}
for fixed $0<c_0<C_0<\infty$.  Put $D=B-r^2/2$.  Then, uniformly in this
range,
\begin{align}
    [\sigma^L]g(\sigma)^{-N}
    &=\frac{(-Nr)^L}{L!}
      \exp\!\left(-\frac{DL^2}{Nr^2}\right)
      \left(1+O(N^{-1/4})\right).
    \label{eq:uniform-extraction}
\end{align}
The same formula holds after multiplication by a holomorphic factor
$q(\sigma)$ with $q(0)=q_0\neq0$, with the additional factor $q_0$ on the
right.
\end{lemma}

\begin{proof}
Set $\rho=L/(Nr)$ and use the circle
\begin{equation}
    \sigma=-\rho e^{i\theta},
    \qquad -\pi\leq\theta\leq\pi.
    \label{eq:saddle-circle}
\end{equation}
Since $\rho=O(N^{-1/2})$, the circle lies inside the holomorphy disc for
large $N$.  Cauchy's formula gives
\begin{equation}
    [\sigma^L]g(\sigma)^{-N}
    =\frac{(-1)^L\rho^{-L}}{2\pi}
      \int_{-\pi}^{\pi}
      \exp\{-iL\theta-N\log g(-\rho e^{i\theta})\}\dd\theta.
    \label{eq:uniform-cauchy-integral}
\end{equation}
Uniformly on this circle,
\begin{equation}
    \log g(\sigma)=r\sigma+D\sigma^2+O(\sigma^3),
\end{equation}
so, after removing the constant $L$, the angular phase is
\begin{equation}
    L(e^{i\theta}-1-i\theta)
    -ND\rho^2e^{2i\theta}+O(N\rho^3).
    \label{eq:uniform-angular-phase}
\end{equation}
Here $N\rho^3=O(N^{-1/2})$ and $ND\rho^2=O(1)$.

On $|\theta|\leq L^{-2/5}$,
\begin{equation}
    L(e^{i\theta}-1-i\theta)
    =-\frac{L\theta^2}{2}+O(L|\theta|^3).
\end{equation}
The quadratic perturbation in \eqref{eq:uniform-angular-phase} may be
replaced by $-ND\rho^2$ with relative error $O(L^{-1/2})$ on the effective
window $|\theta|=O(L^{-1/2})$.  Gaussian integration gives
\begin{equation}
    [\sigma^L]g(\sigma)^{-N}
    =(-1)^L\rho^{-L}\frac{e^L}{\sqrt{2\pi L}}
      e^{-ND\rho^2}
      \left(1+O(L^{-1/2}+N^{-1/2})\right).
\end{equation}
For $L^{-2/5}<|\theta|\leq\theta_0$ the real part drops by at least
$cL\theta^2$, while for $\theta_0\leq|\theta|\leq\pi$ it drops by at least
$cL$.  These arcs are exponentially smaller than the central contribution.
Finally, Stirling's formula and $\rho=L/(Nr)$ give
\begin{equation}
    \rho^{-L}\frac{e^L}{\sqrt{2\pi L}}
    =\frac{(Nr)^L}{L!}\left(1+O(L^{-1})\right),
    \qquad
    ND\rho^2=\frac{DL^2}{Nr^2}.
\end{equation}
This proves \eqref{eq:uniform-extraction}.  Since
$q(-\rho e^{i\theta})=q_0(1+O(N^{-1/2}))$ on the saddle circle, the assertion
with $q$ follows as well.
\end{proof}

\subsection{The multiple-zero action}

On the $\tau$ plane, the cluster principal part has the form
\begin{equation}
    P(\tau^2,z)=\sum_{j=1}^m
      \frac{\alpha_j(\tau)}{z-\beta_j(\tau)},
    \label{eq:multiple-principal-part}
\end{equation}
where
\begin{equation}
    \alpha_j(\tau)
    =a\He_m'(r_j)\tau^{m-1}(1+O(\tau)),
    \qquad
    \beta_j(\tau)
    =z_0+r_j\tau+b\tau^2+O(\tau^3).
    \label{eq:alpha-beta}
\end{equation}
The sum is invariant under $\tau\mapsto-\tau$.  If $u_n$ denotes its
$n$th diagonal coefficient, then exactly
\begin{equation}
    u_n=-\sum_{k=0}^nd_{n-k}[\tau^{2k}]
      \sum_{j=1}^m
      \alpha_j(\tau)\beta_j(\tau)^{-2(n-k)-1}.
    \label{eq:un-exact}
\end{equation}

\begin{proposition}[Multiple-zero coefficient asymptotics]
\label{prop:multiple-asymptotics}
Let $z_0$ have multiplicity $m\geq2$ and let $r_m$ be the largest root of
$\He_m$.  There is a constant $K_{z_0,m}\neq0$ such that
\begin{equation}
    u_n
    =K_{z_0,m}d_nz_0^{-2n}
      n^{-(m-1)/2}e^{r_m\sqrt{2n}}
      \bigl(1+o(1)\bigr).
    \label{eq:multiple-asymptotic}
\end{equation}
Consequently the corresponding cluster Borel transform has radius
$|z_0|^2/2$.
\end{proposition}

\begin{proof}
Fix a positive Hermite root $r$ and substitute $\tau=z_0\sigma$ in its
summand in \eqref{eq:un-exact}.  From \eqref{eq:alpha-beta}, for a constant
$B_r$ depending on the branch,
\begin{equation}
    \frac{\beta_r(z_0\sigma)}{z_0}
    =1+r\sigma+B_r\sigma^2+O(\sigma^3).
    \label{eq:beta-expansion}
\end{equation}
With $p=n-k$ and $L=2k-m+1$, one obtains
\begin{align}
 & [\tau^{2k}]\alpha_r(\tau)\beta_r(\tau)^{-2p-1}
 \nonumber\\
 &\quad =a\He_m'(r)z_0^{-2n+m-2}
 [\sigma^L](1+O(\sigma))
 (1+r\sigma+B_r\sigma^2+O(\sigma^3))^{-2p-1}.
 \label{eq:branch-coefficient}
\end{align}
If $L<0$, the coefficient is exactly zero because the residue starts with
$\tau^{m-1}$; thus only $L\geq0$ remains.  Let $T_{n,k}^{(r)}$ denote the
branch contribution, including the factor $d_p$ and its sign.

For $k$ in a fixed compact multiple of $\sqrt n$, Lemma
\ref{lem:uniform-extraction}, with $N=2p+1$, gives uniformly
\begin{align}
 \eqref{eq:branch-coefficient}
 &=a\He_m'(r)z_0^{-2n+m-2}
 \frac{(-(2p+1)r)^L}{L!}
 \nonumber\\
 &\quad\times
 \exp\!\left[-\frac{(B_r-r^2/2)L^2}{(2p+1)r^2}\right]
 \left(1+O(n^{-1/4})\right).
 \label{eq:branch-uniform}
\end{align}

We next establish the global tail estimates, keeping track of the
double-factorial ratio that determines the correct phase.  For
$0\leq k<n$, the product formula for $d_n$ gives the exact identity
\begin{align}
 \log\frac{d_{n-k}}{d_n}
 &=-k\log(2n)
   -\sum_{j=0}^{k-1}
    \log\left(1-\frac{2j+1}{2n}\right).
    \label{eq:double-factorial-ratio-exact}
\end{align}
Consequently, for every sufficiently small fixed $\varepsilon>0$,
Taylor's formula with a uniform remainder yields
\begin{equation}
 \log\frac{d_{n-k}}{d_n}
 =-k\log(2n)+\frac{k^2}{2n}
   +O\left(\frac{k^3}{n^2}\right),
 \qquad 0\leq k\leq\varepsilon n.
 \label{eq:double-factorial-ratio-uniform}
\end{equation}
Indeed,
$\sum_{j<k}(2j+1)=k^2$, while
$\sum_{j<k}(2j+1)^2=O(k^3)$; choosing $\varepsilon<1/4$, say, keeps every
argument of the logarithm uniformly away from zero.

Fix now $\delta>0$, and first choose $0<\delta_0<\delta$.  Put
\begin{equation}
    g_r(\sigma)
    =1+r\sigma+B_r\sigma^2+O(\sigma^3).
\end{equation}
After decreasing the underlying disc, the Taylor expansion of the logarithm
gives
\begin{equation}
    |\log g_r(\sigma)|\leq(r+\delta_0)|\sigma|.
\end{equation}
Cauchy's formula on
\begin{equation}
    |\sigma|=\frac{L}{(2p+1)(r+\delta_0)}
\end{equation}
for $1\leq L\leq2\varepsilon n+O(1)$ yields, after choosing $\varepsilon$
small enough that the circle stays in the underlying disc and applying
Stirling's formula,
\begin{equation}
 \left|[\sigma^L](1+O(\sigma))
 (1+r\sigma+O(\sigma^2))^{-2p-1}\right|
 \leq C\sqrt{L+1}\frac{((2p+1)(r+\delta_0))^L}{L!}.
 \label{eq:global-cauchy-bound}
\end{equation}
The case $L=0$ satisfies the same bound directly.  Indeed, the maximum-modulus
estimate for $L\geq1$ gives
$C(e(2p+1)(r+\delta_0)/L)^L$, and Stirling's formula produces the displayed
square-root factor.  Its logarithm is $O(\log n)$ throughout this range.
We combine \eqref{eq:double-factorial-ratio-uniform} and
\eqref{eq:global-cauchy-bound}.  Since
$L=2k-m+1$ and $2p+1=2(n-k)+1$, Stirling's formula gives
\begin{align}
 \log L!&=L\log L-L+O(\log n), \\
 \log\frac{2p+1}{2n}
 &=-\frac{k}{n}
   +O\left(\frac1n+\frac{k^2}{n^2}\right).
 \label{eq:phase-elementary-expansions}
\end{align}
Substituting these two relations and
\eqref{eq:double-factorial-ratio-uniform} into the logarithm of the right
side of \eqref{eq:global-cauchy-bound}, and replacing
$L$ by $2k+O(1)$, gives
\begin{align}
 \log\frac{|T_{n,k}^{(r)}|}{d_n|z_0|^{-2n}}
 &\leq
 2k\left(1+\log\frac{(r+\delta_0)\sqrt n}
                         {\sqrt2k}\right)
 +O\left(\frac{k^2}{n}+\frac{k^3}{n^2}+\log n\right).
 \label{eq:phase-upper-preliminary}
\end{align}
The implicit constant is uniform for $k\leq\varepsilon n$.  Choose
$\varepsilon$ sufficiently small, depending on $\delta-\delta_0$.
Then the two error terms containing $k$ are absorbed by
\begin{equation}
 2k\log\frac{r+\delta}{r+\delta_0},
\end{equation}
and we obtain, uniformly for $k\leq\varepsilon n$,
\begin{equation}
 \log\frac{|T_{n,k}^{(r)}|}{d_n|z_0|^{-2n}}
 \leq\sqrt n\,\Phi_{r+\delta}(k/\sqrt n)+O(\log n),
 \label{eq:phase-upper}
\end{equation}
where
\begin{equation}
    \Phi_r(x)=2x\left(1+\log\frac{r}{\sqrt2x}\right).
    \label{eq:phase}
\end{equation}
We set $\Phi_r(0)=0$, its continuous limiting value.
For $k>\varepsilon n$, a fixed Cauchy circle and
$d_{n-k}/d_n\leq1/d_k$ give an $e^{-cn\log n}$ bound after normalization by
$d_n|z_0|^{-2n}$.  Hence no $k$ of order $n$ contributes.

The phase is strictly concave:
\begin{equation}
    \Phi_r'(x)=2\log\frac{r}{\sqrt2x},
    \qquad
    \Phi_r''(x)=-\frac2x.
\end{equation}
It has its unique maximum at $x_r=r/\sqrt2$, with value $r\sqrt2$.
Here is the precise localization consequence.  Choose
$0<a<x_r<b$.  Strict maximality gives a number $c_0>0$ and a sufficiently
large fixed $M>b$ such that
\begin{equation}
 \sup_{x\in[0,a]\cup[b,M]}\Phi_r(x)
 \leq r\sqrt2-3c_0,
 \qquad
 \sup_{x\geq M}\Phi_{r+1}(x)
 \leq r\sqrt2-3c_0.
\end{equation}
By continuity, $\delta>0$ may be chosen so small that the first inequality
remains true with $\Phi_{r+\delta}$ and right side
$r\sqrt2-2c_0$.  The second inequality controls all $x\geq M$, because
$\Phi_{r+\delta}(x)\leq\Phi_{r+1}(x)$ when $\delta<1$.  Therefore
\eqref{eq:phase-upper}, together with the $k>\varepsilon n$ estimate, makes
every $k/\sqrt n\notin[a,b]$ exponentially smaller than
$e^{r\sqrt{2n}}$.  Inside $[a,b]$, strict concavity and
\eqref{eq:branch-uniform} show that
\begin{equation}
    |k-r\sqrt{n/2}|>n^{1/4+\epsilon}
\end{equation}
is smaller by $e^{-cn^{2\epsilon}}$ for any fixed
$0<\epsilon<1/4$.

In the central window write
\begin{equation}
    k=r\sqrt{n/2}+yn^{1/4}.
\end{equation}
For completeness, we now calculate the leading constant rather than only
asserting its nonvanishing.  Set
\begin{equation}
    D_r=B_r-\frac{r^2}{2}.
    \label{eq:Dr-def}
\end{equation}
Using $d_j=2^j\Gamma(j+1/2)/\sqrt\pi$, Stirling's expansion applied to
$d_p/d_n$, $(2p+1)^L$, and $L!$ gives, uniformly for bounded $y$,
\begin{align}
 &\frac{d_p}{d_n}\frac{((2p+1)r)^L}{L!}
 \nonumber\\
 &\quad=
 \frac{2^{-m/2-1/4}}{\sqrt{\pi r}}
 n^{-m/2+1/4}e^{r\sqrt{2n}}
 \exp\left(-\frac{\sqrt2}{r}y^2-\frac{3r^2}{4}\right)
 \bigl(1+o(1)\bigr).
 \label{eq:stirling-central-factor}
\end{align}
One way to verify the constant in \eqref{eq:stirling-central-factor} is to
take logarithms and use
\begin{align}
 \log\Gamma(x+c)
 ={}&(x+c-\tfrac12)\log x-x
    +\tfrac12\log(2\pi)
    +\frac{c(c-1)/2+1/12}{x}
    +O(x^{-2}),
 \label{eq:shifted-stirling}
\end{align}
first with $x=n$ and $x=p=n-k$, and then insert
$k=r\sqrt{n/2}+yn^{1/4}$ and $L=2k-m+1$.  The terms of order
$\sqrt n$ give $r\sqrt{2n}$, those of order $y^2$ give
$-\sqrt2y^2/r$, the logarithmic terms give
$(-m/2+1/4)\log n$, and the remaining constant is
$-(m/2+1/4)\log2-\tfrac12\log(\pi r)-3r^2/4$.

Moreover,
\begin{equation}
 \frac{D_rL^2}{(2p+1)r^2}=D_r+o(1),
 \label{eq:quadratic-saddle-limit}
\end{equation}
and the analytic factor $1+O(\sigma)$ in
\eqref{eq:branch-coefficient} tends to $1$ at the saddle
$\sigma=-L/((2p+1)r)=O(n^{-1/2})$.  Finally, the minus sign in
\eqref{eq:un-exact} and $(-1)^L=(-1)^{m-1}$ produce the total sign
$(-1)^m$.  Hence \eqref{eq:branch-uniform} gives
\begin{align}
 T_{n,k}^{(r)}
 &=D_{z_0,m,r}\,d_nz_0^{-2n}
 n^{-m/2+1/4}e^{r\sqrt{2n}}
 \exp\!\left(-\frac{\sqrt2}{r}y^2\right)(1+o(1)),
 \label{eq:local-discrete-gaussian}
\end{align}
where the constant is explicitly
\begin{equation}
 D_{z_0,m,r}
 =(-1)^m a\He_m'(r)z_0^{m-2}
   \frac{2^{-m/2-1/4}}{\sqrt{\pi r}}
   \exp\left(-\frac{3r^2}{4}-D_r\right).
 \label{eq:multiple-local-constant}
\end{equation}
It is nonzero because $a\neq0$, every Hermite zero is simple, $z_0\neq0$,
and the exponential never vanishes.

To justify summation, fix $0<\eta<1/12$.  The same expansion is uniform for
$|y|\leq n^\eta$ with an envelope $Ce^{-cy^2}$, because
\begin{equation}
 \sqrt n\bigl(\Phi_r(x_r+yn^{-1/4})-\Phi_r(x_r)\bigr)
 =-\frac{\sqrt2}{r}y^2+O(|y|^3n^{-1/4}).
\end{equation}
Strict concavity makes $n^\eta<|y|\leq n^\epsilon$ smaller by
$e^{-cn^{2\eta}}$, while the preceding estimates control the rest.  Thus
the central lattice sum is a dominated Riemann sum.

The branches $r_m$ and $-r_m$ are exchanged by $\tau\mapsto-\tau$.
Because the cluster sum is even in $\tau$, their even-coefficient
contributions are equal and add.  Every other nonzero root satisfies
$|r_j|<r_m$, and \eqref{eq:phase-upper} makes it smaller by
\begin{equation}
    \exp\{-(r_m-|r_j|)\sqrt{2n}+o(\sqrt n)\}.
\end{equation}
If $m$ is odd, the root $r_j=0$ has
$\beta_j(\tau)=z_0+O(\tau^2)$.  We spell out why it is negligible.  Put
$q=(m-1)/2$.  The branch fixed by $\tau\mapsto-\tau$ has the even form
\begin{equation}
 \alpha_0(\tau)=\tau^{2q}A(\tau^2),
 \qquad
 \beta_0(\tau)=z_0B(\tau^2),
 \qquad A(0)\neq0,\quad B(0)=1.
 \label{eq:zero-root-even-form}
\end{equation}
Writing $k=q+h$, a Cauchy radius of order $(h+1)/n$ gives
\begin{equation}
 \left|[u^h]A(u)B(u)^{-2(n-k)-1}\right|
 \leq C_0\frac{(C_1n)^h}{h!}
 \qquad (h\leq\varepsilon n).
\end{equation}
The prefactor covers $h=0$, for which the coefficient is $A(0)$.
Together with
$d_{n-k}/d_n\leq(C/n)^{q+h}$, the normalized sum over this range is at most
\begin{equation}
 Cn^{-q}\sum_{h\geq0}\frac{C^h}{h!}=O(n^{-q}).
\end{equation}
For $h>\varepsilon n$, a fixed Cauchy circle and
$d_{n-k}/d_n\leq1/d_k$ give $e^{-c n\log n}$, exactly as in the preceding
large-$k$ estimate.  Thus the zero-root branch has only polynomial size and
is exponentially negligible compared with $e^{r_m\sqrt{2n}}$.

Finally, the lattice spacing in $y$ is $n^{-1/4}$.  Summing
\eqref{eq:local-discrete-gaussian} contributes
\begin{equation}
    n^{1/4}\int_{\R}e^{-\sqrt2y^2/r_m}\dd y\,(1+o(1)).
\end{equation}
The integral equals $\sqrt{\pi r_m/\sqrt2}$.  Since the two extreme-root
amplitudes are equal, the leading constant in
\eqref{eq:multiple-asymptotic} is
\begin{align}
 K_{z_0,m}
 &=2D_{z_0,m,r_m}\sqrt{\frac{\pi r_m}{\sqrt2}}
 \nonumber\\
 &=(-1)^m2^{(1-m)/2}a\He_m'(r_m)z_0^{m-2}
   \exp\left(-\frac{3r_m^2}{4}-D_{r_m}\right),
 \label{eq:multiple-global-constant}
\end{align}
which is nonzero.  This proves \eqref{eq:multiple-asymptotic}.

Using \eqref{eq:double-factorial-central-binomial}, the normalized Borel
coefficients have the form
\begin{equation}
    K'_{z_0,m}n^{-m/2}e^{r_m\sqrt{2n}}(1+o(1)),
    \qquad K'_{z_0,m}\neq0.
\end{equation}
Their $n$th-root limit is one, which proves the asserted radius.
\end{proof}

\begin{example}[A double-zero example with an explicit cluster]
\label{ex:double-rademacher}
Let $X=\varepsilon_1+\varepsilon_2$, where the $\varepsilon_i$ are
independent symmetric signs.  Then
\begin{equation}
 M(z)=\cosh^2z,
 \qquad
 Q(s,z)=\frac12+\frac12e^{-2s}\cosh(2z).
\end{equation}
At $z_0=i\pi/2$, the MGF has multiplicity two and
\begin{equation}
 M(z_0+w)=-w^2(1+O(w^2)),
\end{equation}
so $m=2$, $a=-1$, and $b=0$.  The exact zero equation near $z_0$ is
\begin{equation}
 \cosh(2w)=e^{2s},
 \qquad
 w=\pm\frac12\operatorname{arcosh}(e^{2s})
   =\pm\bigl(\sqrt{s}+O(s^{3/2})\bigr).
\end{equation}
Choose the branch $w(s)=\sqrt{s}+O(s^{3/2})$ and set
$D(s)=e^{-2s}\sinh(2w(s))$.  Direct differentiation of the exact $Q$ gives
residues $-D(s)$ and $D(s)$ at $z_0+w(s)$ and $z_0-w(s)$, respectively.
The principal part of this single cluster is therefore exactly
\[
 P_{z_0}(s,z)
 =-\frac{D(s)}{z-z_0-w(s)}+\frac{D(s)}{z-z_0+w(s)}
 =-\frac{2D(s)w(s)}{(z-z_0)^2-w(s)^2}.
\]
Both $D(s)w(s)$ and $w(s)^2$ are holomorphic in $s$ near zero, with
$D(s)w(s)=2s+O(s^2)$ and $w(s)^2=s+O(s^2)$.  This directly verifies the
local splitting and the leading residues, independently of the general
coefficient asymptotics.

The splitting roots are the roots $\pm1$ of $\He_2$, while
$B_{r_2}=0$ and $D_{r_2}=-1/2$.  Formula
\eqref{eq:multiple-global-constant} gives
\begin{equation}
 K_{z_0,2}=-\sqrt2\,e^{-1/4}\neq0.
\end{equation}
Consequently the cluster coefficients satisfy
\begin{equation}
 u_n\sim-\sqrt2e^{-1/4}d_nz_0^{-2n}n^{-1/2}e^{\sqrt{2n}}.
\end{equation}
The last two displays are consequences of
Proposition~\ref{prop:multiple-asymptotics}, not an independent verification
of its asymptotic constant.  Here $u_n$ refers only to the cluster at
$z_0$; by symmetry, the cluster at $-z_0$ has the same diagonal series.
\end{example}

\begin{corollary}[Borel radius of a zero cluster]
\label{cor:cluster-radius}
For a zero $z_0$ of any multiplicity $m\geq1$, the Borel transform of its
diagonal cluster has ordinary radius of convergence $|z_0|^2/2$.
\end{corollary}

\begin{proof}
For $m=1$, combine Proposition~\ref{prop:simple-asymptotics} with
\eqref{eq:double-factorial-central-binomial}.  For $m\geq2$, use
Proposition~\ref{prop:multiple-asymptotics}.  The factors
$n^{-m/2}e^{r_m\sqrt{2n}}$ have $n$th-root limit one, so in both cases the
$n$th-root limit of the absolute Borel coefficients is $2/|z_0|^2$.
\end{proof}

\section{Borel continuation and noncancellation}
\label{sec:continuation}

The coefficient estimates in Sections~\ref{sec:simple} and
\ref{sec:multiple} determine the convergence radius of each cluster Borel
transform, but a finite radius alone does not identify its singular point
or exclude cancellation in the full sum.  We address these two issues
separately.  First, continuation away from $z_0^2/2$ identifies that action
as a singularity of the cluster's Taylor branch.  Second, we select a
singular contribution at one action and continue every contribution at a
different action to a regular germ there.  The only clusters that can share
an action are those at $z_0$ and $-z_0$; their possible cancellation is
handled explicitly.

\subsection{Continuation on the Abel double cover}

The deformation principle used below belongs to the classical
Picard--Lefschetz and relative-homology framework \cite{Pham1983}.  Here the
relevant family is an explicit two-sheeted curve.  We give the contour and
endpoint estimates needed for this family, including the behavior at
infinity and at the degenerate fiber $v=0$.

\begin{lemma}[Local Borel continuation]
\label{lem:heat-pole}
Let $P(s,z)$ be the sum of the principal parts belonging to a zero cluster of
$Q(s,\cdot)$ that coalesces at $z_0\neq0$ when $s=0$.  The Borel transform of
the diagonal series generated by $P$ admits analytic continuation along
every piecewise smooth path in $\C\setminus\{0,z_0^2/2\}$ starting at a
nonzero point of a sufficiently small disc about zero, with the initial
branch given by its Taylor series.  No assertion about return to zero on
other branches is needed here.
\end{lemma}

\begin{proof}
\emph{Step 1: the partial Borel transform.}
Expand $P(s,z)=\sum_{k\geq0}P_k(z)s^k$.  Differentiating the contour formula
\eqref{eq:cluster-principal-part} at $s=0$ shows that every $P_k$ is
meromorphic on the Riemann sphere, vanishes at infinity, and has its only
possible pole at $z_0$.  For every compact
$K\Subset\C\setminus\{z_0\}$, Cauchy's estimate in $s$ gives
\begin{equation}
    \sup_{z\in K}|P_k(z)|\leq C_K\rho_K^{-k}.
    \label{eq:Pk-bound}
\end{equation}
Thus
\begin{equation}
    \widetilde P(u,z)=\sum_{k\geq0}P_k(z)\frac{u^k}{k!}
    \label{eq:P-partial-borel}
\end{equation}
is entire in $u$ and holomorphic in $z\neq z_0$, locally uniformly together
with all $u$ derivatives.

We also need uniform control at infinity.  Fix a cluster contour in
\eqref{eq:cluster-principal-part}, and choose $\rho>0$ so that $F$ is
holomorphic on that contour for $|s|\leq\rho$.  Since the contour is bounded,
its integral representation gives $|P(s,z)|\leq C/|z|$ for sufficiently
large $|z|$, uniformly in $|s|\leq\rho$.  Cauchy's estimate then gives
\begin{equation}
 |P_k(z)|\leq \frac{C\rho^{-k}}{|z|},
 \qquad
 |\widetilde P(u,z)|\leq\frac{C e^{|u|/\rho}}{|z|}.
 \label{eq:partial-borel-infinity}
\end{equation}
Thus the partial Borel transform introduces no new spatial singularity
away from $z_0$, although its singularity at $z_0$ need not remain a pole.

\emph{Step 2: the Abel period and its normalization.}
For $v\in\C$, introduce the Abel curve
\begin{equation}
    \Sigma_v=\{(z,y)\in\C^2:y^2=z^2-2v\}.
\end{equation}
For $v\neq0$, the differential $\widetilde P(u,z)\dd z/y$ is regular at
the branch points unless $z=z_0$: the identity $\dd z/y=\dd y/z$ gives
a local expression there.  It is also regular at the two points at infinity
on the compactified curve, since \eqref{eq:partial-borel-infinity} gives
$\widetilde P(u,z)\dd z/y=O(z^{-2})\dd z$.  Consequently, filling in those
two points does not change any of the periods under consideration.

For small nonzero $v$, let $\delta_v$ be the lift, on the sheet
$y\sim z$ at infinity, of a positively oriented loop surrounding the two
branch points $\pm\sqrt{2v}$ and not $z_0$.  Then
\begin{equation}
    \mathcal A(u,v)=\frac1{2\pi i}
    \int_{\delta_v}\widetilde P(u,z)\frac{\dd z}{y}
    \label{eq:abel-cycle}
\end{equation}
agrees with the circular average in \eqref{eq:circular-average}.  At $v=0$
the period extends holomorphically with value $\widetilde P(u,0)$, either by
its circular power series or by shrinking $\delta_v$ to the origin.  The
deleted point $z=z_0$ lies over a branch point precisely when
\begin{equation}
    v=a_0:=\frac{z_0^2}{2}.
\end{equation}

For reference, the equality with the circular average follows from the
explicit parametrization of the lifted cycle
\begin{equation}
 z=\sqrt{2v}\cos\theta,
 \qquad
 y=i\sqrt{2v}\sin\theta,
 \qquad 0\leq\theta\leq2\pi.
 \label{eq:abel-cycle-parametrization}
\end{equation}
Along this parametrization, $\dd z/y=i\,\dd\theta$ after choosing the
orientation of $\delta_v$ consistently.  Hence
\begin{equation}
 \frac1{2\pi i}\int_{\delta_v}G(z)\frac{\dd z}{y}
 =\frac1{2\pi}\int_0^{2\pi}
   G(\sqrt{2v}\cos\theta)\dd\theta.
 \label{eq:abel-circular-identity}
\end{equation}
This calculation also fixes the normalization and orientation used below.

\emph{Step 3: deformation away from the moving collision.}
Before the final $\xi$ derivative in \eqref{eq:double-borel}, the cluster
contribution is
\begin{equation}
    J(\xi)=\xi\int_0^1
    \mathcal A(t\xi,(1-t)\xi)\dd t.
    \label{eq:relative-cycle-parameters}
\end{equation}
This identity initially holds for $|\xi|$ small.  Let
$\gamma:[0,1]\to\C\setminus\{0,a_0\}$ be a piecewise smooth path starting
at a nonzero point of that disc.  For fixed $\xi$, the only value of $t$
for which a branch point of the Abel curve lies over $z=z_0$ is
\begin{equation}
    t_*(\xi)=1-\frac{a_0}{\xi}.
    \label{eq:collision}
\end{equation}
Along $\gamma$ one has $t_*(\xi)\neq0$ because $\xi\neq a_0$, and
$t_*(\xi)\neq1$ because $a_0\neq0$.
Thus the moving collision never reaches either endpoint of the $t$ path.
The separate degeneration at $t=1$, where $v=0$, will be treated in
Step~4 rather than by smooth-fiber transport.

Fix a point $\gamma(\lambda_0)$.  Choose a piecewise smooth path $L$ from
$0$ to $1$ in the complex $t$ plane avoiding
$t_*(\gamma(\lambda_0))$.  After shrinking to a neighborhood $U$ of this
$\xi$ value, the same $L$ avoids $t_*(\xi)$ for every $\xi\in U$.  For
$(t,\xi)\in L\times U$, the parameter $v=(1-t)\xi$ never equals $a_0$.
The fiber $\Sigma_0$ is singular, so one must not invoke a homology local
system across the endpoint $v=0$.  We therefore separate that endpoint.
Since $\gamma$ avoids zero, shrink $U$ so that $0\notin\overline U$, and
choose $L$ so that it meets $t=1$ only at its terminal point.  Parameterize
it by a map $\ell:[0,1]\to\C$ with $\ell(1)=1$.
Choose $\eta>0$ and $v_0<|a_0|/2$ so that
\begin{equation}
 |(1-\ell(\lambda))\xi|<v_0,
 \qquad 1-\eta\leq\lambda\leq1,
 \quad \xi\in\overline U.
 \label{eq:terminal-small-v}
\end{equation}

On the truncated set $0\leq\lambda\leq1-\eta$, the parameter $v$ stays in a
compact subset of $\C\setminus\{0,a_0\}$.  The cycle may therefore be
transported in the smooth punctured family
\begin{equation}
    \Sigma_v\setminus\{z=z_0\},
    \qquad v\neq0,a_0.
\end{equation}
Equivalently, one transports its class in the first-homology local system.
The pullback of this local system to $[0,1-\eta]\times U$ is trivial after
$U$ is chosen contractible.

\emph{Step 4: regularity at the terminal fiber.}
It remains to justify the terminal portion, where $v\to0$.  On the punctured
disc $0<|v|<v_0$, choose a standard vanishing cycle
$\delta_v^{\mathrm{van}}$ around the two branch points.  A transported class
from the truncated portion can be represented, with integer coefficients
independent of $(\lambda,\xi)$, as an integer multiple of this vanishing class and
finitely many small loops supported away from the node $z=0$; the only loops
relevant to the holomorphic form on the punctured curve in
\eqref{eq:abel-cycle} may be taken around
the two points lying over $z=z_0$.  This is the elementary local
decomposition of the homology of a one-node degeneration.

Here is a direct justification of that decomposition.  Choose
$0<\epsilon<|z_0|/3$ and reduce $v_0$ so that $2v_0<\epsilon^2/4$.
The inverse image of $|z|<\epsilon$ is an annulus, whose first homology
is generated by $\delta_v^{\mathrm{van}}$.  On $|z|>\epsilon/2$, including
infinity, the two sheets have the holomorphic expressions
\begin{equation}
 y_\pm(z,v)=\pm z\left(1-\frac{2v}{z^2}\right)^{1/2},
 \label{eq:terminal-outer-sheets}
\end{equation}
where the square root is given by its binomial series.  Each outer sheet,
with its point at infinity filled in, is a disc punctured at the point over
$z_0$.  Its first homology is therefore generated by a small loop about
that puncture.  Gluing these two punctured discs to the inner annulus
shows that the vanishing loop and the two puncture loops generate the
homology needed for the transported period.  Filling in infinity is
legitimate by Step~2.  The generators vary continuously on the terminal
parameter set, so a transported class has an expression with constant
integer coefficients there.

Each resulting period is holomorphic at $v=0$.  For the vanishing part this
follows directly from the convergent circular expansion
\begin{equation}
 \frac1{2\pi}\int_0^{2\pi}
 \widetilde P(u,\sqrt{2v}\cos\theta)\dd\theta
 =\sum_{q\geq0}
   \frac{\partial_z^{2q}\widetilde P(u,0)}{(2q)!}
   \frac1{2\pi}\int_0^{2\pi}
      (\sqrt{2v}\cos\theta)^{2q}\dd\theta,
 \label{eq:terminal-circular-series}
\end{equation}
which is locally uniformly convergent in $(u,v)$.  For a loop over $z_0$,
choose one of the two holomorphic branches of
$y=\sqrt{z^2-2v}$ near $z_0$; this is possible because
$|v|<v_0<|z_0|^2/4$.  Its period is then a fixed small contour integral of a
function holomorphic in $(u,v)$ on the contour, and hence is holomorphic at
$v=0$.  The same argument applies to every loop supported away from the
node.  Thus the transported period on the terminal portion has a
holomorphic extension to $v=0$, regardless of the monodromy accumulated
before entering the small $v$ disc.

\emph{Step 5: holomorphic germs and their compatibility.}
On the truncated parameter set, compactness permits a finite collection of
local contour representatives avoiding $z_0$ and the branch points; the
branch points themselves are regular points of the differential by Step~2.
The representatives may be taken in a bounded $z$ set, since the
differential is regular at infinity.  On each such contour,
\eqref{eq:Pk-bound} gives locally uniform bounds for the partial Borel
series and its parameter derivatives.  On the terminal portion, the
fixed puncture contours and the convergent series
\eqref{eq:terminal-circular-series} give the same bounds, now up to
$v=0$.  These representations show that the transported period is locally
holomorphic in its parameters and justify integration and differentiation
in $\xi$.  Thus
\begin{equation}
    J_U(\xi)=\xi\int_L
    \mathcal A(t\xi,(1-t)\xi)\dd t
\end{equation}
is holomorphic on $U$.  On overlaps, two such germs agree by deforming the
$t$ paths without crossing $t_*(\xi)$ and transporting the Abel cycle in the
same homology class.  More precisely, cover $\gamma$ by a finite ordered
chain of such neighborhoods and choose the path $L$ inductively: on each new
neighborhood it is a small deformation of the path already chosen on the
preceding overlap.  The two paths are then in the same relative homotopy
class in $\C\setminus\{t_*(\xi)\}$, so Cauchy's theorem gives equality of the
two germs.  Paths with a different winding number may produce a different
analytic branch, which is harmless for pathwise analytic continuation.  The
ordered chain therefore continues $J$, and then its $\xi$ derivative, along
the whole path $\gamma$.

In local coordinates about a collision with $\xi_*\neq0$,
$u=t-t_*$, $w=z-z_0$, and $\eta=\xi-\xi_*$, where
$t_*=1-a_0/\xi_*$, the branch locus is
\begin{equation}
    2z_0w+w^2+2\xi_*u-2(1-t_*)\eta+2u\eta=0.
    \label{eq:correct-local-collision}
\end{equation}
Because $z_0\neq0$, this is locally a graph $w=\lambda(u,\eta)$; the
double-cover coordinate is $y$, not $w$.  The homology transport above is
the invariant way to bypass this moving branch point.
\end{proof}

\begin{corollary}[The action is a genuine singularity]
\label{cor:genuine-action}
The Taylor branch of the cluster Borel transform has a nonremovable
singularity at $a_0=z_0^2/2$ and extends holomorphically across every other
point of its convergence circle.
\end{corollary}

\begin{proof}
By Corollary~\ref{cor:cluster-radius}, its Taylor series at zero has finite
radius $|a_0|$.  For a point $b$ on $|\xi|=|a_0|$ with $b\neq a_0$,
start at a sufficiently small nonzero point on the radius towards $b$ and
continue along that radius, including its endpoint.  This path avoids
$0$ and $a_0$, so Lemma~\ref{lem:heat-pole} provides a holomorphic germ at
$b$ agreeing with the original Taylor branch on the inner part of the
radius.  Thus every such $b$ is regular for that branch.
Since a Taylor series with finite radius
has a singular point on its circle of convergence, that point must be
$a_0$.
\end{proof}

\subsection{Distinct zeros cannot cancel}

The action map is
\begin{equation}
    \mathfrak a(z)=\frac{z^2}{2}.
    \label{eq:action-map}
\end{equation}
Two zeros have the same action only when they are $z$ and $-z$.
Accordingly, we first find a singular sum of the clusters sharing one
action, and only then separate that sum from the remaining actions.
This distinction is essential: analytic continuation separates different
actions, but by itself says nothing about cancellation at a shared action.

\begin{proposition}[No cancellation between zero clusters]
\label{prop:no-cancellation}
If $M$ has at least one zero, then the Borel transform of $\widehat A$ has a
nonremovable singularity at a finite action point.
\end{proposition}

\begin{proof}
Write $m(z)$ for the multiplicity of a zero, with $m(z)=0$ if $z$ is not a
zero, and set $B_z=0$ when $m(z)=0$.  We will select a zero $z_0$ for which
\begin{equation}
 S(\xi)=B_{z_0}(\xi)+B_{-z_0}(\xi),
 \qquad a_0=\frac{z_0^2}{2},
 \label{eq:shared-action-sum}
\end{equation}
is holomorphic on $|\xi|<|a_0|$ and singular at $a_0$ on its Taylor
branch.  There are two cases, according to whether the zero divisor is
invariant under reflection.

\emph{Case 1: the zero divisor is not invariant under $z\mapsto-z$.}
Choose $z_0$ with $m(z_0)>m(-z_0)$.  If $m(-z_0)=0$, the required property
of $S$ follows directly from Corollary~\ref{cor:genuine-action}.  Otherwise,
both clusters have the same action scale.  With $r_1:=0$, their
action-normalized Borel coefficients have a nonzero constant prefactor
and the respective scales
\begin{equation}
    n^{-m/2}e^{r_m\sqrt{2n}}
\end{equation}
by Propositions~\ref{prop:simple-asymptotics} and
\ref{prop:multiple-asymptotics}.  Strict interlacing of Hermite zeros makes
$r_m$ strictly increasing in $m$.  Hence the ratio of the smaller-multiplicity
coefficient to the larger-multiplicity coefficient tends to zero, regardless
of their complex prefactors.  The sum $S$ therefore has Taylor radius
$|a_0|$.  Lemma~\ref{lem:heat-pole} applies to both summands away from
$0$ and $a_0$.  The radial argument in
Corollary~\ref{cor:genuine-action} then shows that $a_0$ is the only
singular point on the convergence circle of $S$, as required.

\emph{Case 2: the zero divisor is invariant under $z\mapsto-z$.}
In the Hadamard factorizations of $M(z)$ and $M(-z)$, choose the same
genus-two canonical product.  Pairing opposite zeros makes this product
even, so it cancels in the quotient
\begin{equation}
    R_M(z)=\frac{M(z)}{M(-z)}
    \label{eq:symmetric-divisor-quotient}
\end{equation}
and gives $R_M(z)=\exp(p(z)-p(-z))$ for a polynomial $p$ of degree at
most two.  Thus $R_M(z)=e^{cz}$, and positivity on the real axis gives
$c\in\R$.  Replacing $X$ by $X-c/2$
leaves its MMSE unchanged and preserves a square-exponential moment after
decreasing its exponent if necessary.  The corresponding posterior-energy
function $A(s)=\E X^2-\mmse_X(s)$ changes only by a constant, so its positive
order coefficients and all finite Borel singularities are unchanged.  The
new MGF
\begin{equation}
    \widetilde M(z)=e^{-cz/2}M(z)
\end{equation}
is even.  The exponential factor is zero-free, so the zero divisor and all
action points $z_0^2/2$ are unchanged.  Uniqueness of the MGF makes the
shifted law symmetric.  It follows
that
\begin{equation}
    Q(s,-z)=Q(s,z),
    \qquad F(s,-z)=F(s,z).
    \label{eq:even-QF}
\end{equation}
The principal parts at $z_0$ and $-z_0$ are reflections of one another, and
the circular average \eqref{eq:circular-average} assigns them equal
diagonal series.  In particular, for the shifted law and any zero $z_0$,
$S=2B_{z_0}$ is holomorphic on $|\xi|<|a_0|$ and singular at $a_0$ by
Corollary~\ref{cor:genuine-action}.  Since translation changes the full
Borel transform only by a constant, it suffices to finish the argument
for this law.

\emph{Separation from all other actions.}
In either case we now have a zero $z_0$ and a sum $S$ with the required
Taylor-branch singularity.  Choose $R>|z_0|$ with no zero on $|z|=R$.
Proposition~\ref{prop:localization} gives
\begin{equation}
 \cB\widehat A
 =S
  +\sum_{\substack{|z_j|<R\\z_j\neq\pm z_0}}B_{z_j}+B_H,
 \label{eq:symmetric-localized-decomposition}
\end{equation}
where the sum is finite and $B_H$ is holomorphic on $|\xi|<R^2/2$.
Every action in the remaining sum differs from $a_0$, because the entire
pair $\{z_0,-z_0\}$ has already been included in $S$.

Start at a nonzero point on the radius towards $a_0$, inside the common
initial convergence disc.  Follow that radius towards $a_0$, making small
detours around the finitely many other action points on it.  The detours
can all be taken inside $|\xi|<|a_0|$ and away from zero, and the path
can be radial in a final neighborhood of $a_0$.  Along this path the
terms in \eqref{eq:symmetric-localized-decomposition} behave as follows:
\begin{itemize}
\item The function $S$ remains on its original Taylor branch, since the
entire path before its endpoint lies in its disc of holomorphy.
\item For each remaining cluster, the path including its endpoint avoids
both zero and that cluster's own action.  Lemma~\ref{lem:heat-pole}
therefore supplies a holomorphic germ across $a_0$ on the continued branch.
\item The remainder $B_H$ is already holomorphic across $a_0$, since
$|a_0|<R^2/2$.
\end{itemize}
The identity \eqref{eq:symmetric-localized-decomposition} is preserved
along the path by uniqueness of analytic continuation.  If the continued
full sum admitted a holomorphic extension across $a_0$, subtracting the
regular germs in the last two items would extend the original Taylor
branch of $S$ across $a_0$.  This contradicts the choice of $S$ and proves
the proposition.
\end{proof}

The conclusion is the existence of a finite singularity of a continuation
of the full Borel germ, which is enough to rule out an entire Borel
transform.  The argument does not require the selected action to have
minimal modulus, nor does it assert that every cluster singularity survives
in the full sum.

\section{Completion, consequences, and interpretation}
\label{sec:completion}

\begin{proof}[Proof of Theorem~\ref{thm:main}]
If $X\sim N(\mu,\sigma^2)$, then \eqref{eq:gaussian-mmse} proves analyticity
at zero.

Conversely, suppose $X$ satisfies \eqref{eq:subgaussian} and is
non-Gaussian.  Lemma~\ref{lem:zero-free} gives a complex zero of its MGF.
Proposition~\ref{prop:no-cancellation} then shows that the ordinary Borel
transform of the actual formal expansion $\widehat A$ has a nonremovable
singularity at a finite point.

If $\widehat A(s)=\sum a_ns^n$ had positive ordinary convergence radius,
then $|a_n|\leq CR^{-n}$ for some $C,R>0$, and
\begin{equation}
    \sum_{n\geq0}a_n\frac{\xi^n}{n!}
\end{equation}
would be entire.  This contradicts the finite Borel singularity.  Therefore
$\widehat A$ diverges.  By Proposition~\ref{prop:poincare}, $A$, and hence
$\mmse_X=\E X^2-A$, cannot be analytic at zero.
\end{proof}

\begin{proof}[Proof of Corollary~\ref{cor:rational}]
Since
\[
    \mmse_X(s)\longrightarrow\Var(X)
    \qquad\text{as }s\downarrow0,
\]
the rational function $R$ is bounded on a punctured one-sided neighborhood
of zero.  Hence zero is a removable singularity of $R$.  After filling in
this removable singularity, $R$, and therefore $\mmse_X$, is analytic at
zero.  Theorem~\ref{thm:main} and \eqref{eq:gaussian-mmse} complete the proof.
\end{proof}

\begin{proof}[Proof of Corollary~\ref{cor:mutual-information}]
The I--MMSE identity gives
\begin{equation}
    I_X'(s)=\frac12\mmse_X(s).
\end{equation}
If $I_X$ is analytic at zero, so is its derivative, and
Theorem~\ref{thm:main} implies that $X$ is Gaussian.  The converse is the
standard Gaussian mutual-information formula.
\end{proof}

The argument provides more than a binary analyticity criterion.  It gives a
complex spectrum controlling the large-order zero-SNR coefficients.  For an
individual simple-zero cluster at $z_0$, Proposition
\ref{prop:simple-asymptotics} gives
\begin{equation}
    [s^n]\widehat A_{z_0}(s)
    \sim C_{z_0}(2n-1)!!z_0^{-2n},
    \label{eq:simple-growth-summary}
\end{equation}
with $C_{z_0}\neq0$.  If several zeros have the same minimal action, the
full coefficient is the sum of their cluster contributions, so a separate
noncancellation argument is required.  For an individual zero of
multiplicity $m\geq2$, Hermite splitting adds the factor
\begin{equation}
    n^{-(m-1)/2}e^{r_m\sqrt{2n}}.
    \label{eq:multiple-growth-summary}
\end{equation}
Thus the zero-SNR expansion is naturally Gevrey--1, while the MGF zero set
acts as a complex non-Gaussianity spectrum through
\begin{equation}
    z_0\longmapsto \mathfrak a(z_0)=z_0^2/2.
\end{equation}

The square-exponential moment condition is used twice: to make $M$ an entire
function of order at most two, and to obtain a jointly holomorphic backward
heat evolution with exponentially flat spatial tails.  Determining the
weakest tail condition under which the same rigidity remains true is a
natural open problem.  Other directions include vector Gaussian channels,
matrix-valued SNR, and a full resurgent description of the Stokes data
associated with the MGF zero divisor.

\paragraph{Technical note.}
The most delicate ingredients of the proof are the uniform multiple-zero
saddle analysis in Section~\ref{sec:multiple} and the Abel-double-cover
continuation argument in Section~\ref{sec:continuation}.

\appendix

\section{Detailed estimates for the zero-SNR expansion}
\label{app:poincare-details}

This appendix supplies the uniform estimates used in
Proposition~\ref{prop:poincare}.  The point is to separate the exponentially
small part of the Gaussian integral from the part on which ordinary Taylor
expansion is uniform.

\subsection{Joint holomorphy of the backward heat transform}

Fix $0<\eta<\beta$ and $R<\infty$.  If $|s|\leq\eta$ and $|z|\leq R$,
then
\begin{equation}
    \left|e^{zX-sX^2/2}\right|
    \leq e^{R|X|+\eta X^2/2}.
    \label{eq:appendix-basic-domination}
\end{equation}
Set $\delta=\beta-\eta/2>0$.  Young's inequality gives
$R|x|\leq\delta x^2/2+R^2/(2\delta)$.  For every integer $j\geq0$,
$C_{\delta,j}:=\sup_{x\in\R}|x|^j e^{-\delta x^2/2}<\infty$.
Thus half of the positive margin $\delta$ absorbs the linear exponential,
and the other half absorbs the polynomial factor.  With
$C_{R,\eta,j}=e^{R^2/(2\delta)}C_{\delta,j}$, we obtain
\begin{equation}
    |X|^j e^{R|X|+\eta X^2/2}
    \leq C_{R,\eta,j}e^{\beta X^2},
    \qquad j\geq0.
    \label{eq:appendix-derivative-domination}
\end{equation}
Taking $j=2a+b$ covers every polynomial factor arising from mixed
derivatives, for all nonnegative integers $a,b$.  The right-hand side is
integrable by \eqref{eq:subgaussian}.  Thus every
mixed derivative in $s$ and $z$ may be passed through the expectation,
locally uniformly on $|s|<\beta$, and
\begin{equation}
    \partial_s^a\partial_z^b Q(s,z)
    =\left(-\frac12\right)^a
      \E\!\left[X^{2a+b}e^{zX-sX^2/2}\right].
    \label{eq:appendix-Q-derivatives}
\end{equation}
This proves the joint holomorphy asserted before
\eqref{eq:backward-heat}, rather than merely separate holomorphy.

\subsection{An explicit bound for the outer integral}

Let $K_s(z)=(2\pi s)^{-1/2}e^{-z^2/(2s)}$ for real $s>0$.  From
\eqref{eq:F-domination}--\eqref{eq:complete-square},
\begin{align}
    0
    &\leq \int_{|z|>\delta}K_s(z)F(s,z)\dd z \\
    &\leq \E\!\left[X^2
      \Pp\{|sX+\sqrt{s}G|>\delta\mid X\}\right].
    \label{eq:appendix-outer-start}
\end{align}
Using the union bound from the proof of Proposition~\ref{prop:poincare},
the right-hand side is at most
\begin{equation}
    \E\!\left[X^2\bm 1_{\{|X|>\delta/(2s)\}}\right]
    +\E X^2\,\Pp\{|G|>\delta/(2\sqrt{s})\}.
    \label{eq:appendix-outer-split}
\end{equation}
For the first term, choose $0<\beta_0<\beta$ and write
\begin{align}
    \E\!\left[X^2\bm 1_{\{|X|>\delta/(2s)\}}\right]
    &\leq
      e^{-\beta_0\delta^2/(4s^2)}\E[X^2e^{\beta_0X^2}].
    \label{eq:appendix-X-tail}
\end{align}
For the second, the standard Gaussian estimate gives
\begin{equation}
    \Pp\{|G|>\delta/(2\sqrt{s})\}
    \leq 2e^{-\delta^2/(8s)}.
    \label{eq:appendix-G-tail}
\end{equation}
Consequently, for sufficiently small $s$,
\begin{equation}
    \int_{|z|>\delta}K_s(z)F(s,z)\dd z
    \leq C e^{-c/s}.
    \label{eq:appendix-flat-tail}
\end{equation}
In particular, the outer integral is $O(s^N)$ for every $N$.

\subsection{The parabolic Taylor remainder}

Choose $\delta>0$ and $s_0>0$ such that $F$ is holomorphic on a
neighborhood of the closed bidisc
$|s|\leq s_0$, $|z|\leq2\delta$.  Fix $N\geq0$.  Taylor expansion first in
$s$ gives, uniformly for $|z|\leq\delta$,
\begin{equation}
    F(s,z)=\sum_{k=0}^{N}s^kF_k(z)+s^{N+1}R_N(s,z),
    \qquad |R_N(s,z)|\leq C_N.
    \label{eq:appendix-s-Taylor}
\end{equation}
For each $0\leq k\leq N$, expand $F_k$ to the odd degree
$2(N-k)+1$:
\begin{equation}
    F_k(z)
    =\sum_{\ell=0}^{2(N-k)+1}f_{k,\ell}z^\ell
      +z^{2(N-k)+2}R_{k,N}(z),
    \qquad |R_{k,N}(z)|\leq C_N.
    \label{eq:appendix-z-Taylor}
\end{equation}
The remainder in \eqref{eq:appendix-s-Taylor} contributes
$O(s^{N+1})$ after integration against $K_s$.  The $k$th remainder in
\eqref{eq:appendix-z-Taylor} contributes at most
\begin{align}
    C_Ns^k\int_{\R}K_s(z)|z|^{2(N-k)+2}\dd z
    &=C_Ns^{N+1}\E|G|^{2(N-k)+2} \\
    &=O_N(s^{N+1}).
    \label{eq:appendix-parabolic-remainder}
\end{align}
Replacing the truncated polynomial integrals over $|z|<\delta$ by integrals
over $\R$ changes them by $O(e^{-c/s})$.  Odd powers vanish by symmetry,
and even powers are exactly \eqref{eq:gaussian-moments}.  Hence the retained
terms are precisely
\begin{equation}
    \sum_{k+q\leq N}f_{k,2q}(2q-1)!!\,s^{k+q},
\end{equation}
with a total error $O(s^{N+1})$.  This proves
\eqref{eq:poincare} with all remainders uniform.

\section{Uniform coefficient extraction for a simple zero}
\label{app:simple-details}

We now justify the summation step in
Proposition~\ref{prop:simple-asymptotics}.  The following elementary
large-powers lemma is useful.  It is a parameter-dependent restatement of
standard estimates from singularity analysis and analytic combinatorics
\cite{FlajoletOdlyzko1990,FlajoletSedgewick2009}; no originality is claimed
for the underlying large-powers principle.

\begin{lemma}[Large powers, including the zero-drift case]
\label{lem:analytic-large-powers}
Let $a$ and $h$ be holomorphic on $|s|<\rho$, with
$a(0)=a_0\neq0$, $h(0)=0$, and $h'(0)=c$.  Put
\begin{equation}
    N_{n,k}=2(n-k)+1,
    \qquad 0\leq k\leq n.
\end{equation}
Then, for each fixed $k$,
\begin{equation}
    [s^k]a(s)e^{-N_{n,k}h(s)}
    =a_0\frac{(-2nc)^k}{k!}+O_k(n^{k-1}),
    \label{eq:appendix-fixed-k-large-power}
\end{equation}
where the error for $k=0$ is interpreted as zero.  Fix also $z_0\neq0$.
There are $\varepsilon>0$ and $C_0,C_1>0$ such that
\begin{equation}
    \frac{d_{n-k}}{d_n}|z_0|^{2k}
    \left|[s^k]a(s)e^{-N_{n,k}h(s)}\right|
    \leq C_0\frac{C_1^k}{k!},
    \qquad 0\leq k\leq\varepsilon n,
    \label{eq:appendix-simple-majorant}
\end{equation}
and
\begin{equation}
    \sum_{k>\varepsilon n}\frac{d_{n-k}}{d_n}|z_0|^{2k}
    \left|[s^k]a(s)e^{-N_{n,k}h(s)}\right|=o(1).
    \label{eq:appendix-simple-tail}
\end{equation}
\end{lemma}

\begin{proof}
For fixed $k$, the coefficient $[s^k]a(s)e^{-Nh(s)}$ is a polynomial in
$N$ of degree at most $k$.  Indeed,
\begin{equation}
    e^{-Nh(s)}=\sum_{j=0}^{k}\frac{(-N)^j}{j!}h(s)^j+O(s^{k+1}),
\end{equation}
and only $j\leq k$ can contribute.  Since
$h(s)^k=c^ks^k+O(s^{k+1})$, the coefficient of $N^k$ is
$a_0(-c)^k/k!$.  Substituting $N=N_{n,k}=2n+O_k(1)$ proves
\eqref{eq:appendix-fixed-k-large-power}.  This also covers $c=0$, when the
degree simply drops below $k$.

For uniform bounds, take the Cauchy radius
\begin{equation}
    r_{n,k}=\frac{k+1}{An}
    \label{eq:appendix-simple-cauchy-radius}
\end{equation}
with $A$ large enough that $r_{n,k}<\rho/2$ whenever
$k\leq\varepsilon n$.  Since $h(s)=O(s)$ and $a$ is bounded on the smaller
disc, Cauchy's formula gives
\begin{equation}
    \left|[s^k]a(s)e^{-N_{n,k}h(s)}\right|
    \leq C_0r_{n,k}^{-k}e^{C_0N_{n,k}r_{n,k}}
    \leq C_0\frac{(C_1n)^k}{k!},
    \label{eq:appendix-simple-uniform-bound}
\end{equation}
where the constants may be enlarged in the second inequality, which follows
from Stirling's bound.  The prefactor also covers $k=0$, when the coefficient
is $a_0$.  Also
\begin{equation}
    \frac{d_{n-k}}{d_n}
    =\prod_{j=n-k+1}^{n}\frac1{2j-1}
    \leq \left(\frac{C_2}{n}\right)^k
    \label{eq:appendix-simple-d-ratio-bound}
\end{equation}
uniformly for $k\leq\varepsilon n$ after decreasing $\varepsilon$ if
necessary.  After absorbing $|z_0|^{2k}$ and relabeling the constants this proves
\eqref{eq:appendix-simple-majorant}.

For $k>\varepsilon n$, use a fixed Cauchy radius $r<\rho$.  The coefficient
is bounded by $r^{-k}e^{Cn}$.  Since
\begin{equation}
    \frac{d_{n-k}}{d_n}\leq\frac1{d_k}
    \leq\left(\frac{C_3}{k}\right)^k,
    \label{eq:appendix-simple-large-k}
\end{equation}
the normalized $k$th summand is bounded by
$e^{Cn}(C_4/k)^k$.  Its sum for $k>\varepsilon n$ is
$e^{-\Omega(n\log n)}$, proving \eqref{eq:appendix-simple-tail}.
\end{proof}

Apply Lemma~\ref{lem:analytic-large-powers} with
\begin{equation}
    a(s)=\frac{r(s)}{r(0)},
    \qquad
    h(s)=\log\frac{\zeta(s)}{z_0}.
\end{equation}
After factoring $-r(0)d_nz_0^{-2n-1}$ from
\eqref{eq:simple-diagonal}, the term indexed by $k$ converges to
\begin{equation}
    \frac{(-cz_0^2)^k}{k!}.
    \label{eq:appendix-simple-limiting-term}
\end{equation}
Dominated convergence for the sum over $k$ gives
\begin{equation}
    \sum_{k\geq0}\frac{(-cz_0^2)^k}{k!}
    =e^{-cz_0^2}.
    \label{eq:appendix-simple-exponential-sum}
\end{equation}
Using $c=\zeta'(0)/z_0$ and \eqref{eq:zeta-prime} now gives exactly
\eqref{eq:simple-asymptotic}.

\section{Proof-dependency summary}
\label{app:proof-summary}

For readers checking the argument line by line, the logical dependencies are
as follows.

\begin{enumerate}[label=(\roman*),leftmargin=2.2em]
\item Lemma~\ref{lem:zero-free} turns non-Gaussianity into the existence of
an MGF zero.  Its only external input is Hadamard factorization for an entire
function of finite order.

\item Proposition~\ref{prop:poincare}, with
Appendix~\ref{app:poincare-details}, identifies the formal diagonal series
with the actual right Taylor asymptotics of the MMSE.

\item Proposition~\ref{prop:double-borel} is a coefficient identity.  No
analytic continuation is used in its derivation.

\item Proposition~\ref{prop:localization} ensures that only finitely many
zero clusters matter near a prescribed Borel action.

\item Proposition~\ref{prop:simple-asymptotics} is completed by the
large-powers estimate in Appendix~\ref{app:simple-details}.

\item Proposition~\ref{prop:multiple-asymptotics} follows from the uniform
Cauchy saddle in Lemma~\ref{lem:uniform-extraction}, the global tail bound
\eqref{eq:phase-upper}, the explicit central factor
\eqref{eq:stirling-central-factor}, and the dominated discrete-Gaussian
summation in Section~\ref{sec:multiple}.  The nonzero leading constant is
given in \eqref{eq:multiple-global-constant}.  This is the most delicate
coefficient calculation.

\item Lemma~\ref{lem:heat-pole} transports the Abel cycle in the first
homology local system of the smooth punctured double cover only while
$v\neq0,a_0$.  The singular endpoint $v=0$ is handled separately by the
vanishing-cycle/outer-cycle decomposition and the convergent series
\eqref{eq:terminal-circular-series}.  The continuation paths in the Borel
plane avoid both $0$ and $a_0$; no claim about return to zero on another
branch is used.  The local branch locus is
\eqref{eq:correct-local-collision}, with double-cover coordinate $y$.
This is the most delicate analytic-continuation step.

\item Corollary~\ref{cor:genuine-action} combines the coefficient radius
with Lemma~\ref{lem:heat-pole} to identify $z_0^2/2$ as the unique boundary
obstruction for a zero cluster.

\item Proposition~\ref{prop:no-cancellation} combines finite-disc
localization, local continuation, and the symmetry dichotomy.  Paths to a
target action remain inside its Taylor convergence disc until reaching
the endpoint, so detours around other actions preserve the target's
original singular branch.  The final
contradiction uses only the elementary fact that the Borel transform of a
convergent power series is entire.
\end{enumerate}

\section*{Acknowledgments}

The author used OpenAI Codex through Prism to assist with
language editing, restructuring and drafting portions of the
proof exposition in Section~\ref{sec:continuation}, and checking
cross-references and bibliographic metadata. The author takes
full responsibility for the mathematical claims, proofs, and
references in this article.

\end{document}